\documentclass[apsrev,twocolumn]{revtex4-1}%
\usepackage{amsfonts}
\usepackage{amsmath}
\usepackage{amssymb}
\usepackage{hyperref}
\usepackage{graphicx}
\usepackage{color}%
\providecommand{\U}[1]{\protect\rule{.1in}{.1in}}
\newtheorem{theorem}{Theorem}

\newtheorem{proposition}[theorem]{Proposition}
\newtheorem{remark}[theorem]{Remark}

\newenvironment{proof}[1][Proof]{\noindent\textbf{#1.} }{\ \rule{0.5em}{0.5em}}
\begin{document}
\title{Strong and uniform convergence in the teleportation simulation of\\bosonic Gaussian channels}
\author{Mark M. Wilde}
\affiliation{Hearne Institute for Theoretical Physics, Department of Physics and Astronomy,
Center for Computation and Technology, Louisiana State University, Baton
Rouge, Louisiana 70803, USA}

\begin{abstract}
In the literature on the continuous-variable bosonic teleportation protocol
due to [Braunstein and Kimble, Phys.~Rev.~Lett., 80(4):869, 1998], it is often
loosely stated that this protocol converges to a perfect teleportation of an
input state in the limit of ideal squeezing and ideal detection, but the exact
form of this convergence is typically not clarified. In this paper, I
explicitly clarify that the convergence is in the strong sense, and not the
uniform sense, and furthermore, that the convergence occurs for any input
state to the protocol, including the infinite-energy Basel states defined and
discussed here. I also prove, in contrast to the above result, that the
teleportation simulations of pure-loss, thermal, pure-amplifier, amplifier,
and additive-noise channels converge both strongly and uniformly to the
original channels, in the limit of ideal squeezing and detection for the
simulations. For these channels, I give explicit uniform bounds on the
accuracy of their teleportation simulations. I then extend these uniform convergence results
to particular multi-mode bosonic Gaussian channels. These
convergence statements have important implications for mathematical proofs
that make use of the teleportation simulation of bosonic Gaussian channels,
some of which have to do with bounding their non-asymptotic
secret-key-agreement capacities. As a byproduct of the discussion given here,
I confirm the correctness of the proof of such bounds from my joint work with
Berta and Tomamichel from [Wilde, Tomamichel, Berta, IEEE~Trans.~Inf.~Theory
63(3):1792, March 2017].
Furthermore,
I show that it is not necessary to invoke the energy-constrained diamond
distance in order to confirm the correctness of this proof.

\end{abstract}
\date{\today}
\maketitle


\section{Introduction}

The quantum teleportation protocol is one of the most powerful primitives in
quantum information theory \cite{PhysRevLett.70.1895}. By sharing entanglement
and making use of a classical communication link, a sender can transmit an
arbitrary quantum state to a receiver. In resource-theoretic language, the
resources of a maximally entangled state of two qubits
\begin{equation}
|\Phi^{+}\rangle_{AB}=(|00\rangle_{AB}+|11\rangle_{AB})/\sqrt{2}%
\end{equation}
and two classical bit channels can be used to simulate an ideal qubit channel
from the sender to the receiver \cite{Bennett04,DHW05RI}. Generalizing this, a
maximally entangled state of two qudits
\begin{equation}
|\Phi_{d}\rangle_{AB}=d^{-1/2}\sum_{i=0}^{d-1}|i\rangle_{A}|i\rangle_{B}%
\end{equation}
and two classical channels, each of dimension $d$, can be used to simulate an
ideal $d$-dimensional quantum channel \cite{PhysRevLett.70.1895}. The
teleportation primitive has been extended in multiple non-trivial ways,
including a method to simulate an unideal channel using a noisy, mixed
resource state \cite[Section~V]{BDSW96} (see also
\cite{HHH99,WPG07,NFC09,Mul12}) and as a way to implement nonlocal quantum
gates \cite{GC99,STM11}. The former extension has been used to bound the rates
at which quantum information can be conveyed over a quantum channel assisted
by local operations and classical communication (LOCC) \cite{BDSW96,Mul12},
and more generally, as a way to reduce a general LOCC-assisted protocol to one
that consists of preparing a resource state followed by a single round of
LOCC~\cite{BDSW96,Mul12}.

\begin{figure}[ptb]
\begin{center}
\includegraphics[
width=3in
]{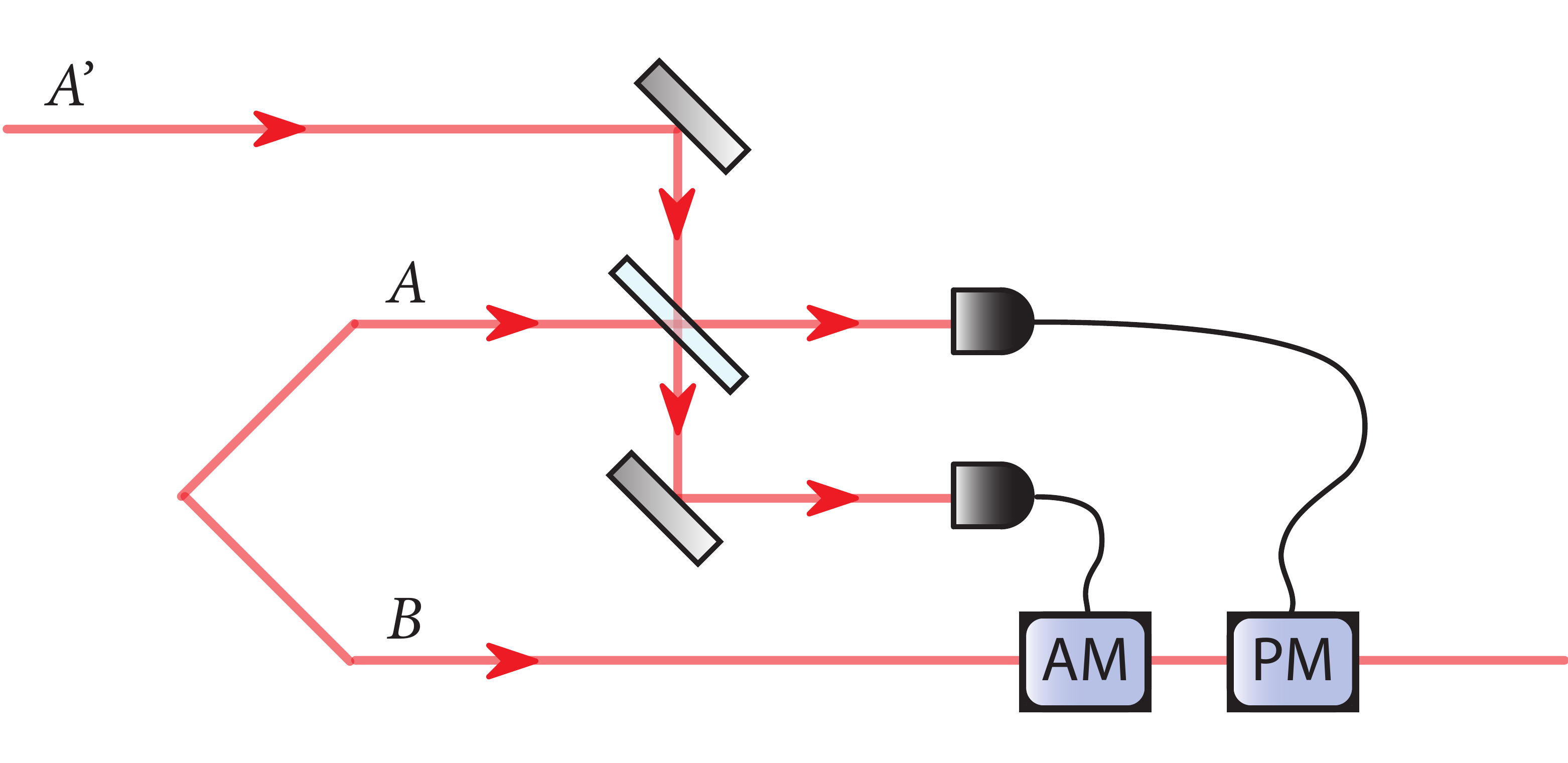}
\end{center}
\caption{Depiction of the bosonic continuous-variable teleportation protocol
from \cite{prl1998braunstein}, described in the main text. ``AM'' and ``PM''
denote amplitude and phase modulators, which implement the displacement
operator needed in the teleportation protocol.}%
\label{fig:CV-tele}\end{figure}

Due in part to the large experimental interest in bosonic continuous-variable
quantum systems, given their practical applications \cite{adesso14,S17}, the
teleportation protocol was extended to this paradigm \cite{prl1998braunstein} (see Figure~\ref{fig:CV-tele}).
The standard protocol begins with a sender and receiver sharing a two-mode
squeezed vacuum state of the following form:%
\begin{equation}
|\Phi(N_{S})\rangle_{AB}\equiv\frac{1}{\sqrt{N_{S}+1}}\sum_{n=0}^{\infty
}\left(  \sqrt{\frac{N_{S}}{N_{S}+1}}\right)  ^{n}|n\rangle_{A}|n\rangle_{B} ,
\label{eq:TMSV}%
\end{equation}
where $N_{S}\in\lbrack0,\infty)$ represents the squeezing strength and
$\{|n\rangle\}_{n}$ denotes the photon-number basis. Suppose that the goal is
to teleport a mode $A^{\prime}$. The sender mixes, on a 50-50 beamsplitter,
the mode $A^{\prime}$ with the mode $A$ of the state in \eqref{eq:TMSV}.
Afterward, the sender performs homodyne detection on the modes emerging from
the beamsplitter and forwards the measurement results over classical channels
to the receiver, who possesses mode $B$ of the above state. Finally, the
receiver performs a unitary displacement operation on mode $B$. In the limit
as $N_{S}\rightarrow\infty$ and in the limit of ideal homodyne detection, this
continuous-variable teleportation protocol is often loosely stated in the
literature to simulate an ideal channel on any state of the mode $A^{\prime}$,
such that this state is prepared in mode $B$ after the protocol is finished.
Due to the lack of a precise notion of convergence being given, there is the
potential for confusion regarding mathematical proofs that make use of the
continuous-variable teleportation protocol.

With this in mind, one purpose of the present paper is to clarify the precise
kind of convergence that occurs in continuous-variable quantum teleportation,
which is typically not discussed in the literature on this topic. In
particular, I prove that the convergence is in the strong topology sense, and
not in the uniform topology sense (see, e.g., \cite[Section~3]{SH08} for
discussions of these notions of convergence). I then show how to extend this
strong convergence result to the teleportation simulation of $n$ parallel
ideal channels, and I also show how this strong convergence extends to the
teleportation simulation of $n$ ideal channels that could be used in any context.

Strong convergence and uniform convergence are then discussed for the
teleportation simulation of bosonic Gaussian channels. For this latter case,
and in contrast to the result discussed above for the continuous-variable
teleportation protocol, I prove that the teleportation simulations of the
pure-loss, thermal, pure-amplifier, amplifier, and additive-noise channels
converge both strongly and uniformly to the original channels, in the limit of
ideal squeezing and detection for the simulations. Here I give explicit
uniform bounds on the accuracy of the teleportation simulations of these
channels, and I suspect that these bounds will be useful in future applications.
After this development, I then extend these uniform convergence results to particular multi-mode bosonic Gaussian channels.

These convergence results are important, even if they might be implicit in
prior works, as they provide meaningful clarification of mathematical proofs
that make use of teleportation simulation, such as those given in recent work
on bounding non-asymptotic secret-key-agreement capacities. In particular, one
can employ these convergence statements to confirm the correctness of the
proof of such bounds given in my joint work with Berta and Tomamichel from
\cite{WTB16}.
Furthermore, these strong convergence statements can be used to conclude that
the energy-constrained diamond distance is \textit{not necessary} to arrive at
a proof of the bounds from \cite{WTB16}.
Another byproduct of the discussion given in the present paper is that it is
clarified that the methods of \cite{WTB16} allow for bounding secret-key rates
of rather general protocols that make use of infinite-energy states, such as
the Basel states in \eqref{eq:entangled-Basel-state} and
\eqref{eq:basel-state-sep}. Although there should be great skepticism
concerning whether these infinite-energy Basel states could be generated in
practice, this latter byproduct is nevertheless of theoretical interest.

The rest of the paper proceeds as follows. In the next section, I discuss the
precise form of convergence that occurs in continuous-variable quantum
teleportation and then develop various extensions of this notion of
convergence. I then prove that teleportation simulations of the pure-loss,
thermal, pure-amplifier, amplifier, and additive-noise channels converge both
strongly and uniformly to the original channels in the limit of ideal
squeezing and detection for the simulations. The uniform convergence results are then extended to the teleportation simulations of particular multi-mode bosonic Gaussian channels.  Section~\ref{sec:TP-game} gives a
physical interpretation of the aforementioned convergence results, by means of
the CV Teleportation Game. After that, Section~\ref{sec:non-asymp-SKC}%
\ briefly reviews what is meant by a secret-key-agreement protocol and
non-asymptotic secret-key-agreement capacity. Finally, in
Section~\ref{sec:proof-review}, I review the proof of \cite[Theorem~24]{WTB16}
and carefully go through some of its steps therein, confirming its
correctness, while showing how the strong convergence of teleportation
simulation applies. In Section~\ref{sec:conclusion}, I conclude with a brief
summary and a discussion.

\section{Notions of quantum channel convergence, with applications to
teleportation simulation}

One main technical issue discussed in this paper is how the
continuous-variable bosonic teleportation protocol from
\cite{prl1998braunstein}\ converges to an identity channel in the limit of
infinite squeezing and ideal detection. This issue is often not explicitly
clarified in the literature on the topic, even though it has been implicit for
some time in various works that the convergence is to be understood in the
strong sense (topology of strong convergence), and \textit{not} necessarily
the uniform sense (topology of uniform convergence) (see, e.g.,
\cite[Section~3]{SH08}). For example, in the original paper
\cite{prl1998braunstein}, the following statement is given regarding this issue:

\begin{quote}
\textquotedblleft Clearly, for $r\rightarrow\infty$ the teleported state of
Eq.~(4) reproduces the original unknown state.\textquotedblright
\end{quote}

\noindent Although it is clear that this statement implies convergence in the
strong sense, it could be helpful to clarify this point,
and the purpose of this section is to do so.

In what follows, I first recall the definitions of strong and uniform
convergence from \cite[Section~3]{SH08}. I then discuss the precise form of
convergence that occurs in continuous-variable bosonic teleportation and show
how strong convergence and uniform convergence are extremely different in the
setting of continuous-variable teleportation. After that, I prove that strong
convergence of a channel sequence implies strong convergence of $n$-fold
tensor powers of these channels and follow this with a proof that strong
convergence of a channel sequence implies strong convergence of $n$ uses of
these channels in any context in which they could be invoked. I also prove
that the teleportation simulations of pure-loss, thermal, pure-amplifier,
amplifier, and additive-noise channels converge both strongly and uniformly to
the original channels, in the limit of ideal squeezing and detection for the simulations. The uniform convergence results are then extended to the teleportation simulations of particular multi-mode bosonic Gaussian channels.

\subsection{Definitions of strong and uniform convergence}

Before discussing the precise statement of convergence in the
continuous-variable bosonic teleportation protocol, let us begin by recalling
general definitions of strong and uniform convergence from \cite[Section~3]%
{SH08}. I adopt slightly different definitions from those given in
\cite[Section~3]{SH08}, in order to suit the needs of the present paper, but
note that they are equivalent to the original definitions as shown in
\cite{SH08} and \cite[Lemma~2]{Sh17}. In this context, see also \cite{Sh17a}.
Let $\{\mathcal{N}_{A\rightarrow B}^{k}\}_{k}$ denote a sequence of quantum
channels (completely positive, trace-preserving maps), which each accept as
input a trace-class operator acting on a separable Hilbert space
$\mathcal{H}_{A}$\ and output a trace-class operator acting on a separable
Hilbert space $\mathcal{H}_{B}$. This sequence converges \textit{strongly} to
a channel $\mathcal{N}_{A\rightarrow B}$ if for all density operators
$\rho_{RA}$ acting on $\mathcal{H}_{R}\otimes\mathcal{H}_{A}$, where
$\mathcal{H}_{R}$ is an arbitrary, auxiliary separable Hilbert space, the
following limit holds%
\begin{equation}
\lim_{k\rightarrow\infty}\varepsilon(k,\rho_{RA})=0,
\end{equation}
where the infidelity is defined as%
\begin{multline}
\varepsilon(k,\rho_{RA})\equiv\\
1-F((\operatorname{id}_{R}\otimes\mathcal{N}_{A\rightarrow B}^{k})(\rho
_{RA}),(\operatorname{id}_{R}\otimes\mathcal{N}_{A\rightarrow B})(\rho_{RA})),
\end{multline}
$\operatorname{id}_{R}$ denotes the identity map on the auxiliary space, and
$F(\tau,\omega)\equiv\left\Vert \sqrt{\tau}\sqrt{\omega}\right\Vert _{1}^{2}$
is the quantum fidelity \cite{U76}, defined for density operators $\tau$ and
$\omega$. The quantum fidelity obeys a data processing inequality, which is
the statement that%
\begin{equation}
F(\mathcal{M}(\tau),\mathcal{M}(\omega))\geq F(\tau,\omega),
\end{equation}
for states $\tau$ and $\omega$ and a quantum channel $\mathcal{M}$. We can
summarize strong convergence more compactly as the following mathematical
statement:%
\begin{equation}
\sup_{\rho_{RA}}\lim_{k\rightarrow\infty}\varepsilon(k,\rho_{RA})=0.
\label{eq:strong-conv-def-compact}%
\end{equation}
Due to purification, the Schmidt decomposition theorem, and the data
processing inequality for fidelity, we find that for every mixed state
$\rho_{RA}$, there exists a pure state $\psi_{R^{\prime}A}$ with the auxiliary
Hilbert space $\mathcal{H}_{R^{\prime}}$ taken to be isomorphic to
$\mathcal{H}_{A}$, such that%
\begin{equation}
\varepsilon(k,\rho_{RA})\leq\varepsilon(k,\psi_{R^{\prime}A}).
\end{equation}
Thus, when considering strong convergence, it suffices to consider only pure
states $\psi_{R^{\prime}A}$, so that%
\begin{equation}
\sup_{\rho_{RA}}\lim_{k\rightarrow\infty}\varepsilon(k,\rho_{RA})=\sup
_{\psi_{R^{\prime}A}}\lim_{k\rightarrow\infty}\varepsilon(k,\psi_{R^{\prime}%
A}). \label{eq:pure-state-mixed-state}%
\end{equation}

Strong convergence is strictly different from uniform convergence
\cite[Section~3]{SH08}, which amounts to a swap of the supremum and the limit
in \eqref{eq:strong-conv-def-compact}. That is, the channel sequence
$\{\mathcal{N}_{A\rightarrow B}^{k}\}_{k}$ converges \textit{uniformly} to the
channel $\mathcal{N}_{A\rightarrow B}$ if the following holds%
\begin{equation}
\lim_{k\rightarrow\infty}\sup_{\rho_{RA}}\varepsilon(k,\rho_{RA})=0.
\end{equation}
Even though this swap might seem harmless and is of no consequence in finite
dimensions, the issue is important to consider in infinite-dimensional
contexts, especially for bosonic channels. That is, a sequence of channels
could converge in the strong sense, but be as far as possible from converging
in the uniform sense, and an example of this behavior is given in the next subsection.

It was also stressed in \cite{SH08} that the topology of uniform convergence
is too strong for physical applications and should typically not be
considered.

\subsection{Strong and uniform convergence considerations for
continuous-variable teleportation}

We now turn our attention to convergence in the continuous-variable bosonic
teleportation protocol and focus on \cite[Eq.~(9)]{prl1998braunstein}, which
states that an unideal continuous-variable bosonic teleportation protocol with
input mode$~A$ realizes the following additive-noise quantum Gaussian channel
$\mathcal{T}_{A}^{\bar{\sigma}}$ on an input density operator~$\rho_{A}$:%
\begin{equation}
\rho_{A}\rightarrow\mathcal{T}_{A}^{\bar{\sigma}}(\rho_{A})\equiv\int
d^{2}\alpha\ G_{\bar{\sigma}}(\alpha)\ D(\alpha)\rho_{A}D(-\alpha),
\label{eq:TP-as-add-noise}%
\end{equation}
where $D(\alpha)$ is a displacement operator \cite{S17} and%
\begin{equation}
G_{\bar{\sigma}}(\alpha)\equiv\frac{1}{\pi\bar{\sigma}}\exp\!\left(
-\frac{\left\vert \alpha\right\vert ^{2}}{\bar{\sigma}}\right)
\label{eq:complex-Gauss}%
\end{equation}
is a zero-mean, circularly symmetric complex Gaussian probability density
function with variance $\bar{\sigma}>0$. To be clear, the integral in
\eqref{eq:TP-as-add-noise} is over the whole complex plane~$\alpha
\in\mathbb{C}$. For an explicit proof of \eqref{eq:TP-as-add-noise}, one can
also consult \cite{TBS02,BST02}. The variance parameter $\bar{\sigma}$
quantifies unideal squeezing and unideal detection. Thus, for any $\bar
{\sigma}>0$, the teleportation channel $\mathcal{T}_{A}^{\bar{\sigma}}$ is
unideal and intuitively becomes ideal in the limit $\bar{\sigma}\rightarrow0$.
However, it is this convergence that needs to be made precise. To examine
this, we need a measure of the channel input-output dissimilarity, and the
entanglement infidelity is a good choice, which is essentially the choice made
in \cite{prl1998braunstein} for quantifying the performance of unideal bosonic
teleportation. For a fixed pure state $\psi_{RA}\equiv|\psi\rangle\langle
\psi|_{RA}$ of modes $R$ and $A$, the entanglement infidelity of the channel
$\mathcal{T}_{A}^{\bar{\sigma}}$ with respect to $\psi_{RA}$\ is defined as%
\begin{equation}
\varepsilon(\bar{\sigma},\psi_{RA})\equiv1-\langle\psi|_{RA}(\operatorname{id}%
_{R}\otimes\mathcal{T}_{A}^{\bar{\sigma}})(|\psi\rangle\langle\psi|_{RA}%
)|\psi\rangle_{RA}. \label{eq:ent-infidel}%
\end{equation}
Examining \cite[Eq. (11)]{prl1998braunstein}, we see that the entanglement
infidelity can alternatively be written as%
\begin{equation}
\varepsilon(\bar{\sigma},\psi_{RA})=1-\int d^{2}\alpha\ G_{\bar{\sigma}%
}(\alpha)\left\vert \chi_{\psi_{A}}(\alpha)\right\vert ^{2},
\end{equation}
where $\chi_{\psi_{A}}(\alpha)=\operatorname{Tr}\{D(\alpha)\psi_{A}\}$ is the
Wigner characteristic function of the reduced density operator $\psi_{A}$. By
applying the H\"{o}lder inequality, we conclude that $\chi_{\psi_{A}}(\alpha)$
is bounded for all $\alpha\in\mathbb{C}$ because%
\begin{equation}
\left\vert \operatorname{Tr}\{D(\alpha)\psi_{A}\}\right\vert \leq\left\Vert
D(\alpha)\right\Vert _{\infty}\left\Vert \psi_{A}\right\Vert _{1}=1.
\end{equation}
Exploiting the continuity of $\chi_{\psi_{A}}(\alpha)$ at $\alpha=0$ and the
fact that $\chi_{\psi_{A}}(0)=1$ \cite[Theorem~5.4.1]{Hol11}, as well as
invoking the boundedness of $\chi_{\psi_{A}}(\alpha)$ and \cite[Theorem~9.8]%
{WZ77} regarding the convergence of nascent delta functions, we then conclude
that for a given state $\psi_{RA}$, the following strong convergence holds%
\begin{equation}
\lim_{\bar{\sigma}\rightarrow0}\varepsilon(\bar{\sigma},\psi_{RA})=0,
\label{eq:fixed-state-converge}%
\end{equation}
which can be written, as before, more compactly as%
\begin{equation}
\sup_{\psi_{RA}}\lim_{\bar{\sigma}\rightarrow0}\varepsilon(\bar{\sigma}%
,\psi_{RA})=0. \label{eq:TP-infidelity}%
\end{equation}
Note that, as before and due to \eqref{eq:pure-state-mixed-state},
Eq.~\eqref{eq:TP-infidelity} implies that%
\begin{equation}
\sup_{\rho_{RA}}\lim_{\bar{\sigma}\rightarrow0}\varepsilon(\bar{\sigma}%
,\rho_{RA})=0
\end{equation}
for any mixed state $\rho_{RA}$ where%
\begin{equation}
\varepsilon(\bar{\sigma},\rho_{RA})\equiv1-F(\rho_{RA},(\operatorname{id}%
_{R}\otimes\mathcal{T}_{A}^{\bar{\sigma}})(\rho_{RA}))
\end{equation}
and $\mathcal{H}_{R}$ is an arbitrary auxiliary separable Hilbert space.

One should note here that the convergence in \eqref{eq:TP-infidelity} already
calls into question any claim regarding the necessity of an energy constraint
for the states that are to be teleported using the continuous-variable
teleportation protocol.
Clearly, the state $\psi_{RA}$ to be teleported could be chosen as the
following Basel state:%
\begin{equation}
|\beta\rangle_{RA}=\sqrt{\frac{6}{\pi^{2}}}\sum_{n=1}^{\infty}\sqrt{\frac
{1}{n^{2}}}|n\rangle_{R}|n\rangle_{A}, \label{eq:entangled-Basel-state}%
\end{equation}
which has mean photon number equal to $\infty$, but it also
satisfies~\eqref{eq:fixed-state-converge}. Such a state is called a
\textquotedblleft Basel state,\textquotedblright\ due to its normalization
factor being connected with the well known Basel problem, which establishes
that $\sum_{n=1}^{\infty}1/n^{2}=\pi^{2}/6$. For $\hat{n}=\sum_{n=0}^{\infty
}n|n\rangle\langle n|$ the photon-number operator, one can easily check that
the mean photon number $\operatorname{Tr}\{(\hat{n}_{R}+\hat{n}_{A})\beta
_{RA}\}=\infty$, due to the presence of the divergent harmonic series after
$\hat{n}_{A}$ multiplies the reduced density operator $\beta_{A}$. Thus, the
only constraint needed for the convergence in \eqref{eq:TP-infidelity} is that
the state to be teleported be a state (i.e., normalizable). Ref.~\cite{P17} claims that it is necessary for there to be an energy constraint for strong convergence in teleportation; however, the example of the Basel states given above proves that such an energy constraint is not necessary. 

It is also important to note that an exchange of the limit and the supremum in
\eqref{eq:TP-infidelity}\ leads to a drastically different conclusion:%
\begin{equation}
\lim_{\bar{\sigma}\rightarrow0}\sup_{\psi_{RA}}\varepsilon(\bar{\sigma}%
,\psi_{RA})=1. \label{eq:bad-converge}%
\end{equation}
The drastic difference is due to the fact that for all fixed $\bar{\sigma}>0$,
one can find a sequence of states $\{\psi_{RA}^{k}\}_{k}$ such that $\sup
_{k}\varepsilon(\bar{\sigma},\psi_{RA}^{k})=1$, establishing
\eqref{eq:bad-converge}. For example, one could pick each $\psi_{RA}^{k}$ to
be a two-mode squeezed vacuum state with squeezing parameter increasing with
increasing $k$. In the limit of large squeezing, the ideal channel and the
additive-noise channel for any $\bar{\sigma}>0$ become perfectly
distinguishable, having infidelity approaching one, implying
\eqref{eq:bad-converge}. One can directly verify this calculation by employing
the covariance matrix representation of the two-mode squeezed vacuum and the
additive-noise channel, as well as the overlap formula in \cite[Eq.~(4.51)]%
{S17}\ to calculate entanglement fidelity.

I now give details of the aforementioned calculation, regarding how the
continuous-variable bosonic teleportation protocol from
\cite{prl1998braunstein} does not converge uniformly to an ideal channel.
Consider a two-mode squeezed vacuum state $\Phi(N_{S})$ with mean photon
number $N_{S}$ for one of its reduced modes, as defined in \eqref{eq:TMSV}.
Such a state has a Wigner-function covariance matrix \cite{S17}\ as follows:%
\begin{multline}%
\begin{bmatrix}
2N_{S}+1 & 2\sqrt{N_{S}(N_{S}+1)}\\
2\sqrt{N_{S}(N_{S}+1)} & 2N_{S}+1
\end{bmatrix}
\\
\oplus%
\begin{bmatrix}
2N_{S}+1 & -2\sqrt{N_{S}(N_{S}+1)}\\
-2\sqrt{N_{S}(N_{S}+1)} & 2N_{S}+1
\end{bmatrix}
.
\end{multline}
After sending one mode of this state through an additive-noise channel with
variance $\bar{\sigma}$ (corresponding to an unideal continuous-variable
bosonic teleportation), the covariance matrix becomes as follows,
corresponding to a state $\tau(N_{S},\bar{\sigma})$:%
\begin{multline}%
\begin{bmatrix}
2N_{S}+1 & 2\sqrt{N_{S}(N_{S}+1)}\\
2\sqrt{N_{S}(N_{S}+1)} & 2N_{S}+1+2 \bar{\sigma}%
\end{bmatrix}
\\
\oplus%
\begin{bmatrix}
2N_{S}+1 & -2\sqrt{N_{S}(N_{S}+1)}\\
-2\sqrt{N_{S}(N_{S}+1)} & 2N_{S}+1+2 \bar{\sigma}%
\end{bmatrix}
.
\end{multline}
The overlap $\operatorname{Tr}\{\omega\sigma\}$ of two zero-mean, two-mode
Gaussian states $\omega$ and $\sigma$ is given by \cite[Eq.~(4.51)]{S17}%
\begin{equation}
\operatorname{Tr}\{\omega\sigma\}=4/\sqrt{\det(V_{\omega}+V_{\sigma})},
\end{equation}
where $V_{\omega}$ and $V_{\sigma}$ are the Wigner-function covariance
matrices of $\omega$ and $\sigma$, respectively. We can then employ this
formula to calculate the overlap%
\begin{equation}
\langle\Phi(N_{S})|\tau(N_{S},\bar{\sigma})|\Phi(N_{S})\rangle
=\operatorname{Tr}\{\Phi(N_{S})\tau(N_{S},\bar{\sigma})\}
\end{equation}
as%
\begin{equation}
\langle\Phi(N_{S})|\tau(N_{S},\bar{\sigma})|\Phi(N_{S})\rangle=\frac{1}%
{\bar{\sigma}+2\bar{\sigma}N_{S}+1}.
\end{equation}
Thus, for a fixed $\bar{\sigma}>0$ and in the limit as $N_{S}\rightarrow
\infty$, we find that $1-\langle\Phi(N_{S})|\tau(N_{S},\bar{\sigma}%
)|\Phi(N_{S})\rangle\rightarrow1$, so that the continuous-variable bosonic
teleportation protocol from \cite{prl1998braunstein}\ does not converge
uniformly to an ideal channel.

To summarize, the kind of convergence considered in \eqref{eq:TP-infidelity}
is the strong sense (topology of strong convergence), whereas the kind of
convergence considered in \eqref{eq:bad-converge}\ is the uniform sense
(topology of uniform convergence) (see, e.g., \cite[Section~3]{SH08}). That
is, \eqref{eq:TP-infidelity} demonstrates that unideal continuous-variable
bosonic teleportation converges \textit{strongly} to an ideal quantum channel
in the limit of ideal squeezing and detection, whereas \eqref{eq:bad-converge}
demonstrates that it does \textit{not} converge \textit{uniformly}.

\subsection{Strong and uniform convergence for tensor-power channels}

A natural consideration to make in the context of quantum Shannon theory is
the convergence of $n$ uses of a channel on a general state of $n$ systems,
where $n$ is a positive integer. To this end, suppose that the strong
convergence in \eqref{eq:strong-conv-def-compact}\ holds for the sequence
$\{\mathcal{N}_{A\rightarrow B}^{k}\}_{k}$ of channels. Then it immediately
follows that the sequence $\{(\mathcal{N}_{A\rightarrow B}^{k})^{\otimes
n}\}_{k}$ converges strongly to $\mathcal{N}_{A\rightarrow B}^{\otimes n}$.
Indeed, for an arbitrary density operator $\rho_{RA^{n}}$ acting on
$\mathcal{H}_{R}\otimes\mathcal{H}_{A}^{\otimes n}$, we are now interested in
bounding the infidelity for the tensor-power channels $(\mathcal{N}%
_{A\rightarrow B}^{k})^{\otimes n}$ and $\mathcal{N}_{A\rightarrow B}^{\otimes
n}$:%
\begin{multline}
\varepsilon^{(n)}(k,\rho_{RA^{n}})\equiv\\
1-F((\operatorname{id}_{R}\otimes(\mathcal{N}_{A\rightarrow B}^{k})^{\otimes
n})(\rho_{RA^{n}}),(\operatorname{id}_{R}\otimes\mathcal{N}_{A\rightarrow
B}^{\otimes n})(\rho_{RA^{n}})),
\end{multline}
in the limit as $k\rightarrow\infty$. By employing the fact that
\begin{equation}
P(\tau,\omega)\equiv\sqrt{1-F(\tau,\omega)}%
\end{equation}
obeys the triangle inequality \cite{R02,R03,GLN04,R06}, we conclude that, for
an arbitrary density operator $\rho_{RA^{n}}$, the following inequality holds
\begin{widetext}
\begin{align}
& P\left[  (\operatorname{id}_{R}\otimes(\mathcal{N}_{A\rightarrow B}%
^{k})^{\otimes n})(\rho_{RA^{n}}),(\operatorname{id}_{R}\otimes\mathcal{N}%
_{A\rightarrow B}^{\otimes n})(\rho_{RA^{n}})\right]
\nonumber\\
& \leq\sum_{i=1}^{n}P\left[  (\operatorname{id}_{R}\otimes(\mathcal{N}%
_{A\rightarrow B}^{k})^{\otimes i}\otimes\mathcal{N}_{A\rightarrow B}^{\otimes
n-i})(\rho_{RA^{n}}),(\operatorname{id}_{R}\otimes(\mathcal{N}_{A\rightarrow
B}^{k})^{\otimes i-1}\otimes\mathcal{N}_{A\rightarrow B}^{\otimes
n-i+1})(\rho_{RA^{n}})\right]
\label{eq:tensor-power-bnd-1}\\
& \leq\sum_{i=1}^{n}P\left[  (\operatorname{id}_{R}\otimes\operatorname{id}%
_{A}^{\otimes  i-1  }\otimes\mathcal{N}_{A\rightarrow B}%
^{k}\otimes\mathcal{N}_{A\rightarrow B}^{\otimes n-i})(\rho_{RA^{n}%
}),(\operatorname{id}_{R}\otimes\operatorname{id}_{A}^{\otimes
i-1  }\otimes\mathcal{N}_{A\rightarrow B}\otimes\mathcal{N}%
_{A\rightarrow B}^{\otimes n-i})(\rho_{RA^{n}})\right]
.\label{eq:tensor-power-bnd-3}%
\end{align}
\end{widetext}The first inequality follows from the triangle inequality for
$P(\tau,\omega)$, and the second follows from data processing for the fidelity
under the channel%
\begin{equation}
\operatorname{id}_{R}\otimes(\mathcal{N}_{A\rightarrow B}^{k})^{\otimes
i-1}\otimes\operatorname{id}_{A}^{\otimes n-i+1}%
\end{equation}
acting on the states%
\begin{align}
&  (\operatorname{id}_{R}\otimes\operatorname{id}_{A}^{\otimes i-1}%
\otimes\mathcal{N}_{A\rightarrow B}^{k}\otimes\mathcal{N}_{A\rightarrow
B}^{\otimes n-i})(\rho_{RA^{n}}),\\
&  (\operatorname{id}_{R}\otimes\operatorname{id}_{A}^{\otimes i-1}%
\otimes\mathcal{N}_{A\rightarrow B}\otimes\mathcal{N}_{A\rightarrow
B}^{\otimes n-i})(\rho_{RA^{n}}).
\end{align}
The method used in
\eqref{eq:tensor-power-bnd-1}--\eqref{eq:tensor-power-bnd-3} is related to the
telescoping approach of \cite{LS09}, employed in the context of continuity of
quantum channel capacities (see the proof of \cite[Theorem~11]{LS09} in
particular). Now employing strong convergence of the channel sequence
$\{\mathcal{N}_{A\rightarrow B}^{k}\}_{k}$ and the fact that $\mathcal{N}%
_{A\rightarrow B}^{\otimes n-i}(\rho_{RA^{n}})$ is a \textit{fixed state}
independent of $k$ for each $i\in\{1,\ldots,n\}$, we conclude that for all
density operators $\rho_{RA^{n}}$%
\begin{equation}
\lim_{k\rightarrow\infty}\varepsilon^{(n)}(k,\rho_{RA^{n}})=0.
\end{equation}
This result can be summarized more compactly as%
\begin{multline}
\sup_{\rho_{RA}}\lim_{k\rightarrow\infty}\varepsilon(k,\rho_{RA}%
)=0\label{eq:str-conv-tensor-power}\\
\Longrightarrow\sup_{\rho_{RA^{n}}}\lim_{k\rightarrow\infty}\varepsilon
^{(n)}(k,\rho_{RA^{n}})=0.
\end{multline}
That is, the strong convergence of the channel sequence $\{\mathcal{N}%
_{A\rightarrow B}^{k}\}_{k}$ implies the strong convergence of the
tensor-power channel sequence $\{(\mathcal{N}_{A\rightarrow B}^{k})^{\otimes
n}\}_{k}$ for any finite $n$.

For the case of continuous-variable teleportation, we have that $n$ unideal
teleportations with the same performance corresponds to the tensor-power
channel $(\mathcal{T}_{A}^{\bar{\sigma}})^{\otimes n}$. By appealing to
\eqref{eq:TP-infidelity} and \eqref{eq:str-conv-tensor-power},\ or
alternatively employing Wigner characteristic functions, \cite[Theorem~5.4.1]%
{Hol11}, and \cite[Theorem~9.8]{WZ77}, the following convergence holds: for a
pure state $\psi_{RA^{n}}$, we have that%
\begin{equation}
\lim_{\bar{\sigma}\rightarrow0}\varepsilon^{(n)}(\bar{\sigma},\psi_{RA^{n}%
})=0, \label{eq:n-converge-tele}%
\end{equation}
where%
\begin{multline}
\varepsilon^{(n)}(\bar{\sigma},\psi_{RA^{n}})\equiv\\
1-\langle\psi|_{RA^{n}}(\operatorname{id}_{R}\otimes(\mathcal{T}_{A}%
^{\bar{\sigma}})^{\otimes n})(|\psi\rangle\langle\psi|_{RA^{n}})|\psi
\rangle_{RA^{n}}.
\end{multline}
Again, more compactly, this is the same as%
\begin{equation}
\sup_{\psi_{RA^{n}}}\lim_{\bar{\sigma}\rightarrow0}\varepsilon^{(n)}%
(\bar{\sigma},\psi_{RA^{n}})=0.
\end{equation}
and drastically different from
\begin{equation}
\lim_{\bar{\sigma}\rightarrow0}\sup_{\psi_{RA^{n}}}\varepsilon^{(n)}%
(\bar{\sigma},\psi_{RA^{n}})=1.
\end{equation}

I end this subsection by noting the following proposition, having to do with
the strong convergence of parallel compositions of strongly converging channel
sequences. In fact, the parallel composition result in
\eqref{eq:str-conv-tensor-power} could be proven by using only the following
proposition and iterating.

\begin{proposition}
\label{prop:parallel-comp} Let $\{\mathcal{N}_{A_{1}\rightarrow B_{1}}%
^{k}\}_{k}$ be a channel sequence that converges strongly to a channel
$\mathcal{N}_{A_{1}\rightarrow B_{1}}$, and let $\{\mathcal{M}_{A_{2}%
\rightarrow B_{2}}^{k}\}_{k}$ be a channel sequence that converges strongly to
a channel $\mathcal{M}_{A_{2}\rightarrow B_{2}}$. Then the channel sequence
$\{\mathcal{N}_{A_{1}\rightarrow B_{1}}^{k}\otimes\mathcal{M}_{A_{2}%
\rightarrow B_{2}}^{k}\}_{k}$ converges strongly to $\mathcal{N}%
_{A_{1}\rightarrow B_{1}}\otimes\mathcal{M}_{A_{2}\rightarrow B_{2}}$.
\end{proposition}

\begin{proof}
A proof is similar to what is given above, and I give it for completeness. Let
$\rho_{RA_{1}A_{2}}$ be an arbitrary state. Consider that\begin{widetext}%
\begin{align}
& \!\!\!\!\!\!\!\! P((\operatorname{id}_{R}\otimes\mathcal{N}_{A_{1}\rightarrow B_{1}}%
^{k}\otimes\mathcal{M}_{A_{2}\rightarrow B_{2}}^{k})(\rho_{RA_{1}A_{2}%
}),(\operatorname{id}_{R}\otimes\mathcal{N}_{A_{1}\rightarrow B_{1}}%
\otimes\mathcal{M}_{A_{2}\rightarrow B_{2}})(\rho_{RA_{1}A_{2}}))\nonumber\\
&  \leq P((\operatorname{id}_{R}\otimes\mathcal{N}_{A_{1}\rightarrow B_{1}%
}^{k}\otimes\mathcal{M}_{A_{2}\rightarrow B_{2}}^{k})(\rho_{RA_{1}A_{2}%
}),(\operatorname{id}_{R}\otimes\mathcal{N}_{A_{1}\rightarrow B_{1}}%
^{k}\otimes\mathcal{M}_{A_{2}\rightarrow B_{2}})(\rho_{RA_{1}A_{2}%
}))\nonumber\\
&  \qquad+P((\operatorname{id}_{R}\otimes\mathcal{N}_{A_{1}\rightarrow B_{1}%
}^{k}\otimes\mathcal{M}_{A_{2}\rightarrow B_{2}})(\rho_{RA_{1}A_{2}%
}),(\operatorname{id}_{R}\otimes\mathcal{N}_{A_{1}\rightarrow B_{1}}%
\otimes\mathcal{M}_{A_{2}\rightarrow B_{2}})(\rho_{RA_{1}A_{2}}))\nonumber\\
&  \leq P((\operatorname{id}_{RA_{1}}\otimes\mathcal{M}_{A_{2}\rightarrow
B_{2}}^{k})(\rho_{RA_{1}A_{2}}),\mathcal{M}_{A_{2}\rightarrow B_{2}}%
(\rho_{RA_{1}A_{2}}))\nonumber\\
&  \qquad+P((\operatorname{id}_{R}\otimes\mathcal{N}_{A_{1}\rightarrow B_{1}%
}^{k}\otimes\operatorname{id}_{A_{2}})(\rho_{RA_{1}A_{2}}),(\operatorname{id}%
_{R}\otimes\mathcal{N}_{A_{1}\rightarrow B_{1}}\otimes\operatorname{id}%
_{A_{2}})(\rho_{RA_{1}A_{2}})).\label{eq:parallel-str-conv}%
\end{align}
The first inequality follows from the triangle inequality and the second from
data processing. Using the strong convergence of $\{\mathcal{N}_{A_{1}%
\rightarrow B_{1}}^{k}\}_{k}$ and $\{\mathcal{M}_{A_{2}\rightarrow B_{2}}%
^{k}\}_{k}$, applying the inequality in \eqref{eq:parallel-str-conv}, and
taking the limit $k\rightarrow\infty$, we find that%
\begin{equation}
\lim_{k\rightarrow\infty}P((\mathcal{N}_{A_{1}\rightarrow B_{1}}^{k}%
\otimes\mathcal{M}_{A_{2}\rightarrow B_{2}}^{k})(\rho_{RA_{1}A_{2}%
}),(\mathcal{N}_{A_{1}\rightarrow B_{1}}\otimes\mathcal{M}_{A_{2}\rightarrow
B_{2}})(\rho_{RA_{1}A_{2}}))=0.
\end{equation}
\end{widetext}Since the state $\rho_{RA_{1}A_{2}}$ was arbitrary, the proof is complete.
\end{proof}

\subsection{Strong and uniform convergence in arbitrary contexts}

\label{sec:strong-arbitrary}The most general way to distinguish $n$ uses of
two different quantum channels is by means of an adaptive protocol. Such
adaptive channel discrimination protocols have been considered extensively in
the literature in the context of finite-dimensional quantum channel
discrimination (see, e.g., \cite{WY06,CDP08,DFY09,HHLW10,CMW14,TW16}).
However, to the best of my knowledge, the issues of strong and uniform
convergence have not yet been considered explicitly in the literature in the
context of infinite-dimensional channel discrimination using adaptive
strategies. The purpose of this section is to clarify these issues by defining
strong and uniform convergence in this general context and then to prove
explicitly that strong convergence of a channel sequence $\{\mathcal{N}%
_{A\rightarrow B}^{k}\}_{k}$ to $\mathcal{N}_{A\rightarrow B}$ implies strong
convergence of $n$ uses of each channel in $\{\mathcal{N}_{A\rightarrow B}%
^{k}\}_{k}$ to $n$ uses of $\mathcal{N}_{A\rightarrow B}$ in the rather
general sense described below.

To clarify what is meant by an adaptive protocol for channel discrimination,
suppose that the task is to distinguish $n$ uses of the channel $\mathcal{N}%
_{A\rightarrow B}^{k}$ from $n$ uses of the channel $\mathcal{N}_{A\rightarrow
B}$. The most general protocol for doing so begins with the preparation of a
state $\rho_{R_{1}A_{1}}$, where system $A_{1}$ is isomorphic to the channel
input system $A$ and $R_{1}$ corresponds to an arbitrary auxiliary separable
Hilbert space. The system $A_{1}$ is then fed in to the first channel use of
$\mathcal{N}_{A\rightarrow B}^{k}$ or $\mathcal{N}_{A\rightarrow B}$,
depending on which of these channels is chosen from the start. The resulting
state is then either%
\begin{equation}
\mathcal{N}_{A\rightarrow B}^{k}(\rho_{R_{1}A_{1}}) \quad\text{ or }
\quad\mathcal{N}_{A\rightarrow B}(\rho_{R_{1}A_{1}}),
\end{equation}
depending on which channel is selected, and where I have omitted the identity
map on $R_{1}$ for simplicity. After this, the discriminator applies a quantum
channel $\mathcal{A}_{R_{1}B_{1}\rightarrow R_{2}A_{2}}^{(1)}$, where $R_{2}$
corresponds to another arbitrary separable Hilbert space, which need not be
isomorphic to $R_{1}$, and $A_{2}$ corresponds to a separable Hilbert space
isomorphic to the channel input $A$. The discriminator then calls the second
use of $\mathcal{N}_{A\rightarrow B}^{k}$ or $\mathcal{N}_{A\rightarrow B}$,
such that the state is now either%
\begin{equation}
(\mathcal{N}_{A_{2}\rightarrow B_{2}}^{k}\circ\mathcal{A}_{R_{1}%
B_{1}\rightarrow R_{2}A_{2}}^{(1)}\circ\mathcal{N}_{A_{1}\rightarrow B_{1}%
}^{k})(\rho_{R_{1}A_{1}}),
\end{equation}
or%
\begin{equation}
(\mathcal{N}_{A_{2}\rightarrow B_{2}}\circ\mathcal{A}_{R_{1}B_{1}\rightarrow
R_{2}A_{2}}^{(1)}\circ\mathcal{N}_{A_{1}\rightarrow B_{1}})(\rho_{R_{1}A_{1}%
}).
\end{equation}
This process continues for $n$ channel uses, and then the final state is
either%
\begin{multline}
\omega_{R_{n}B_{n}}^{k}\equiv\label{eq:omega-k}\\
\left(  \mathcal{N}_{A_{n}\rightarrow B_{n}}^{k}\circ\left[  \bigcirc
_{j=1}^{n-1}\mathcal{A}_{R_{j}B_{j}\rightarrow R_{j+1}A_{j+1}}^{(j)}%
\circ\mathcal{N}_{A_{j}\rightarrow B_{j}}^{k}\right]  \right)  (\rho
_{R_{1}A_{1}}),
\end{multline}
or%
\begin{multline}
\omega_{R_{n}B_{n}}\equiv\label{eq:omega}\\
\left(  \mathcal{N}_{A_{n}\rightarrow B_{n}}\circ\left[  \bigcirc_{j=1}%
^{n-1}\mathcal{A}_{R_{j}B_{j}\rightarrow R_{j+1}A_{j+1}}^{(j)}\circ
\mathcal{N}_{A_{j}\rightarrow B_{j}}\right]  \right)  (\rho_{R_{1}A_{1}}).
\end{multline}
Let $\mathcal{P}^{(n)}$ denote the full protocol, which consists of the state
preparation $\rho_{R_{1}A_{1}}$ and the $n-1$ channels $\{\mathcal{A}%
_{R_{j}B_{j}\rightarrow R_{j+1}A_{j+1}}^{(j)}\}_{j=1}^{n-1}$. The infidelity
in this case, for the fixed protocol $\mathcal{P}^{(n)}$, is then equal to%
\begin{equation}
\varepsilon_{\text{ad}}^{(n)}(k,\mathcal{P}^{(n)})\equiv1-F(\omega_{R_{n}%
B_{n}}^{k},\omega_{R_{n}B_{n}}). \label{eq:adap-err}%
\end{equation}
Figure~\ref{fig:adapt}\ depicts these channels and states, which are used in a
general adaptive strategy to discriminate three uses of $\mathcal{N}%
_{A\rightarrow B}^{k}$ from three uses of $\mathcal{N}_{A\rightarrow B}$.

\begin{figure*}[ptb]
\begin{center}
\includegraphics[
width=4.3927in
]{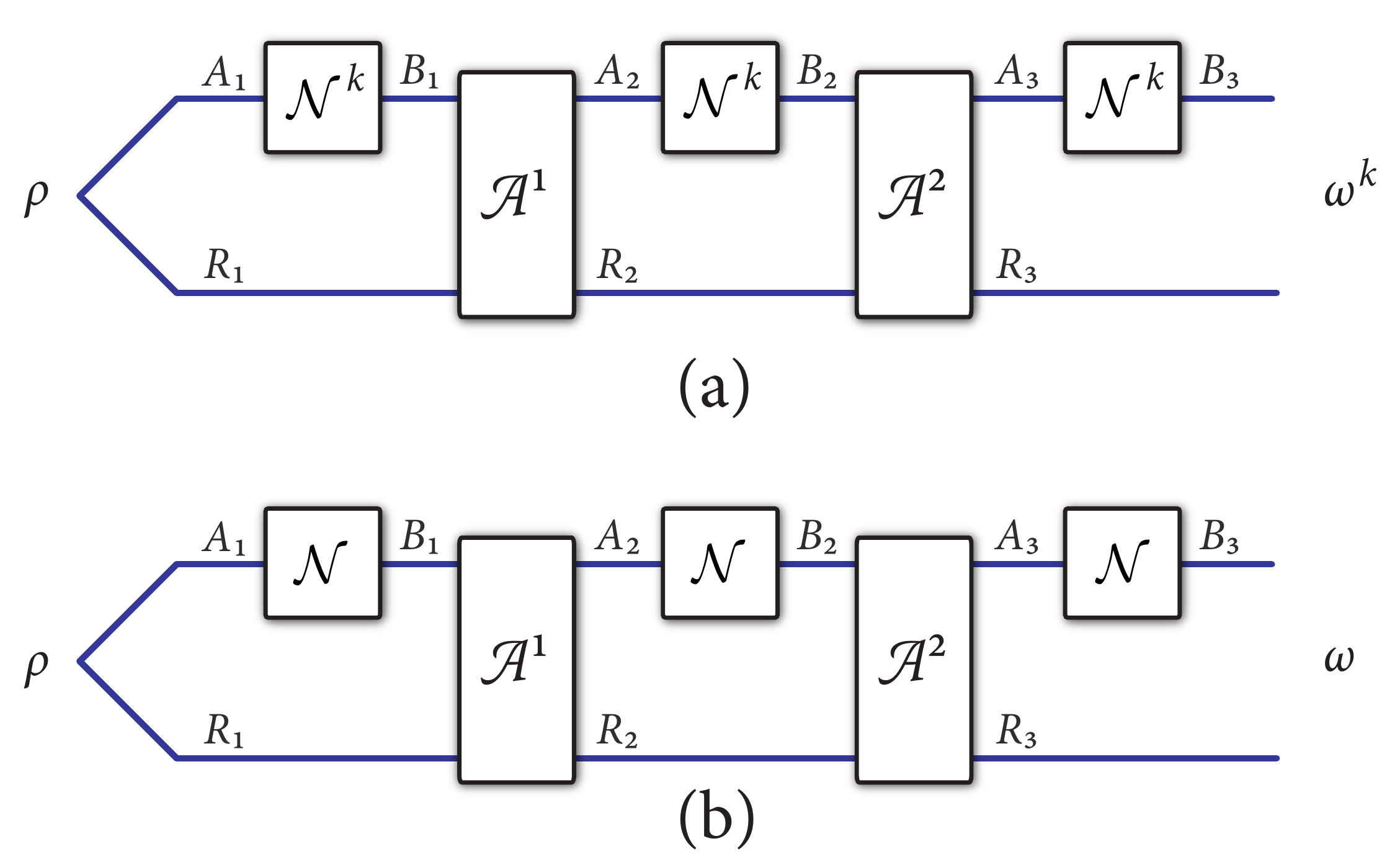}
\end{center}
\caption{Adaptive protocol for distinguishing three uses of the channel
$\mathcal{N}_{A\rightarrow B}^{k}$ from three uses of $\mathcal{N}%
_{A\rightarrow B}$. The protocol is denoted by $\mathcal{P}^{(3)}$ and
consists of state preparation $\rho_{R_{1}A_{1}}$, as well as the channels
$\mathcal{A}_{R_{1}B_{1}\rightarrow R_{2}A_{2}}^{(1)}$ and $\mathcal{A}%
_{R_{2}B_{2}\rightarrow R_{3}A_{3}}^{(2)}$. (a) The protocol $\mathcal{P}%
^{(3)}$ is used with three uses of the channel $\mathcal{N}_{A\rightarrow
B}^{k}$, and the final state is $\omega_{R_{3}B_{3}}^{k}$. (b) The protocol
$\mathcal{P}^{(3)}$ is used with three uses of the channel $\mathcal{N}%
_{A\rightarrow B}$, and the final state is $\omega_{R_{3}B_{3}}$. Strong
convergence of the channel sequence $\{\mathcal{N}_{A\rightarrow B}^{k}\}_{k}$
to $\mathcal{N}_{A\rightarrow B}$ implies that, given a fixed protocol
$\mathcal{P}^{(3)}$, the infidelity of the states $\omega_{R_{3}B_{3}}^{k}$
and $\omega_{R_{3}B_{3}}$ converges to zero in the limit as $k\rightarrow
\infty$.}%
\label{fig:adapt}%
\end{figure*}

In this general context, strong convergence corresponds to the following
statement: for a given protocol $\mathcal{P}^{(n)}$, the following limit holds%
\begin{equation}
\lim_{k\rightarrow\infty}\varepsilon_{\text{ad}}^{(n)}(k,\mathcal{P}^{(n)})=0,
\end{equation}
or more compactly,%
\begin{equation}
\sup_{\mathcal{P}^{(n)}}\lim_{k\rightarrow\infty}\varepsilon_{\text{ad}}%
^{(n)}(k,\mathcal{P}^{(n)})=0.
\end{equation}
Uniform convergence again corresponds to a swap of the supremum and limit%
\begin{equation}
\lim_{k\rightarrow\infty}\sup_{\mathcal{P}^{(n)}}\varepsilon_{\text{ad}}%
^{(n)}(k,\mathcal{P}^{(n)})=0,
\end{equation}
and again, it should typically be avoided in physical applications as it is
too strong and not needed for most purposes, following the suggestions of
\cite[Section~3]{SH08}.

I now explicitly show that strong convergence of the sequence $\{\mathcal{N}%
_{A\rightarrow B}^{k}\}_{k}$ implies strong convergence of $n$ uses of each
channel in this sequence in this general sense. The proof is elementary and
similar to that in
\eqref{eq:tensor-power-bnd-1}--\eqref{eq:tensor-power-bnd-3}, making use of
the triangle inequality and data processing of fidelity. It bears similarities
to \textit{numerous} prior results in the literature
\cite{BBBV97,BBHT98,Z99,BHLS03,CMW14,TGW14,TGW14Nat,Christandl2017}, in which
adaptive protocols were analyzed. For simplicity, we can focus on the case of
$n=3$ and then the proof is easily extended. Begin by considering a
\textit{fixed protocol} $\mathcal{P}^{(3)}$. Then consider
that\begin{widetext}%
\begin{align}
\sqrt{\varepsilon_{\text{ad}}^{(3)}(k,\mathcal{P}^{(3)})} &  =P\left[
\omega_{R_{3}B_{3}}^{k},\omega_{R_{3}B_{3}}\right]  \nonumber\\
&  =P\left[  \left(  \mathcal{N}^{k}\circ\mathcal{A}^{(2)}\circ\mathcal{N}%
^{k}\circ\mathcal{A}^{(1)}\circ\mathcal{N}^{k}\right)  (\rho_{R_{1}A_{1}%
}),\left(  \mathcal{N}\circ\mathcal{A}^{(2)}\circ\mathcal{N}\circ
\mathcal{A}^{(1)}\circ\mathcal{N}\right)  (\rho_{R_{1}A_{1}})\right]
\nonumber\\
&  \leq P\left[  \left(  \mathcal{N}^{k}\circ\mathcal{A}^{(2)}\circ
\mathcal{N}^{k}\circ\mathcal{A}^{(1)}\circ\mathcal{N}^{k}\right)  (\rho
_{R_{1}A_{1}}),\left(  \mathcal{N}^{k}\circ\mathcal{A}^{(2)}\circ
\mathcal{N}^{k}\circ\mathcal{A}^{(1)}\circ\mathcal{N}\right)  (\rho
_{R_{1}A_{1}})\right]  \nonumber\\
&  \qquad+P\left[  \left(  \mathcal{N}^{k}\circ\mathcal{A}^{(2)}%
\circ\mathcal{N}^{k}\circ\mathcal{A}^{(1)}\circ\mathcal{N}\right)
(\rho_{R_{1}A_{1}}),\left(  \mathcal{N}^{k}\circ\mathcal{A}^{(2)}%
\circ\mathcal{N}\circ\mathcal{A}^{(1)}\circ\mathcal{N}\right)  (\rho
_{R_{1}A_{1}})\right]  \nonumber\\
&  \qquad+P\left[  \left(  \mathcal{N}^{k}\circ\mathcal{A}^{(2)}%
\circ\mathcal{N}\circ\mathcal{A}^{(1)}\circ\mathcal{N}\right)  (\rho
_{R_{1}A_{1}}),\left(  \mathcal{N}\circ\mathcal{A}^{(2)}\circ\mathcal{N}%
\circ\mathcal{A}^{(1)}\circ\mathcal{N}\right)  (\rho_{R_{1}A_{1}})\right]
\nonumber\\
&  \leq P\left[  \mathcal{N}^{k}(\rho_{R_{1}A_{1}}),\mathcal{N}(\rho
_{R_{1}A_{1}})\right]  \nonumber\\
&  \qquad+P\left[  \mathcal{N}^{k}\left[  (\mathcal{A}^{(1)}\circ
\mathcal{N})(\rho_{R_{1}A_{1}})\right]  ,\mathcal{N}\left[  (\mathcal{A}%
^{(1)}\circ\mathcal{N})(\rho_{R_{1}A_{1}})\right]  \right]  \nonumber\\
&  \qquad+P\left[  \mathcal{N}^{k}\left[  (\mathcal{A}^{(2)}\circ
\mathcal{N}\circ\mathcal{A}^{(1)}\circ\mathcal{N})(\rho_{R_{1}A_{1}})\right]
,\mathcal{N}\left[  (\mathcal{A}^{(2)}\circ\mathcal{N}\circ\mathcal{A}%
^{(1)}\circ\mathcal{N})(\rho_{R_{1}A_{1}})\right]  \right]
,\label{eq:adap-err-dist}%
\end{align}
\end{widetext}where I have omitted some system labels for simplicity. The
first inequality follows from the triangle inequality and the second from data
processing of the fidelity under the channels%
\begin{align}
&  \mathcal{N}^{k}\circ\mathcal{A}^{(2)}\circ\mathcal{N}^{k}\circ
\mathcal{A}^{(1)},\text{ and}\\
&  \mathcal{N}^{k}\circ\mathcal{A}^{(2)}.
\end{align}
The inequality in \eqref{eq:adap-err-dist}\ can be understood as saying that
the overall distinguishability of the $n$ uses of $\mathcal{N}^{k}$ and
$\mathcal{N}$, as captured by $P\left[  \omega_{R_{3}B_{3}}^{k},\omega
_{R_{3}B_{3}}\right]  $, is limited by the sum of the distinguishabilities at
every step in the discrimination protocol (this is similar to the observations
made in \cite{BBBV97,BBHT98,Z99,BHLS03,CMW14,TGW14,TGW14Nat,Christandl2017}).
Now employing the inequality in \eqref{eq:adap-err-dist}, the facts that%
\begin{align}
&  \rho_{R_{1}A_{1}},\\
&  (\mathcal{A}^{(1)}\circ\mathcal{N})(\rho_{R_{1}A_{1}}),\text{ and}\\
&  (\mathcal{A}^{(2)}\circ\mathcal{N}\circ\mathcal{A}^{(1)}\circ
\mathcal{N})(\rho_{R_{1}A_{1}})
\end{align}
are \textit{fixed states} independent of $k$, the strong convergence of
$\{\mathcal{N}_{A\rightarrow B}^{k}\}_{k}$, and taking the limit
$k\rightarrow\infty$ on both sides of the inequality in
\eqref{eq:adap-err-dist}, we conclude that for any fixed protocol
$\mathcal{P}^{(3)}$, the following limit holds%
\begin{equation}
\lim_{k\rightarrow\infty}\varepsilon_{\text{ad}}^{(3)}(k,\mathcal{P}^{(3)})=0.
\end{equation}
By the same reasoning with the triangle inequality and data processing, the
argument extends to any finite positive integer $n$, so that for any fixed
protocol $\mathcal{P}^{(n)}$, the following limit holds%
\begin{equation}
\lim_{k\rightarrow\infty}\varepsilon_{\text{ad}}^{(n)}(k,\mathcal{P}^{(n)})=0.
\end{equation}
We can summarize the above development more compactly as%
\begin{multline}
\sup_{\rho_{RA}}\lim_{k\rightarrow\infty}\varepsilon(k,\rho_{RA}%
)=0\label{eq:strong-conv-arbitrary}\\
\Longrightarrow\sup_{\mathcal{P}^{(n)}}\lim_{k\rightarrow\infty}%
\varepsilon_{\text{ad}}^{(n)}(k,\mathcal{P}^{(n)})=0.
\end{multline}
That is, strong convergence of the channel sequence $\{\mathcal{N}%
_{A\rightarrow B}^{k}\}_{k}$ implies strong convergence of $n$ uses of each
channel $\mathcal{N}_{A\rightarrow B}^{k}$ in this sequence in any context in
which the $n$ uses of $\mathcal{N}_{A\rightarrow B}^{k}$ could be invoked.

I end this subsection by noting the following proposition, having to do with
the strong convergence of serial compositions of channel sequences. In fact,
the serial composition result in \eqref{eq:strong-conv-arbitrary} for adaptive
protocols could be proven by employing only the following proposition and iterating.

\begin{proposition}
\label{prop:sequential-comp} Let $\{\mathcal{N}_{A\rightarrow B}^{k}\}_{k}$ be
a channel sequence that converges strongly to a channel $\mathcal{N}%
_{A\rightarrow B}$, and let $\{\mathcal{M}_{B\rightarrow C}^{k}\}_{k}$ be a
channel sequence that converges strongly to a channel $\mathcal{M}%
_{B\rightarrow C}$. Then the channel sequence $\{\mathcal{M}_{B\rightarrow
C}^{k}\circ\mathcal{N}_{A\rightarrow B}^{k}\}_{k}$ converges strongly to
$\mathcal{M}_{B\rightarrow C}\circ\mathcal{N}_{A\rightarrow B}$.
\end{proposition}

\begin{proof}
A proof is similar to what is given above, and I give it for completeness. Let
$\rho_{RA}$ be an arbitrary state. Consider that%
\begin{align}
&  P((\mathcal{M}_{B\rightarrow C}^{k}\circ\mathcal{N}_{A\rightarrow B}%
^{k})(\rho_{RA}),(\mathcal{M}_{B\rightarrow C}\circ\mathcal{N}_{A\rightarrow
B})(\rho_{RA}))\nonumber\\
&  \leq P((\mathcal{M}_{B\rightarrow C}^{k}\circ\mathcal{N}_{A\rightarrow
B}^{k})(\rho_{RA}),(\mathcal{M}_{B\rightarrow C}^{k}\circ\mathcal{N}%
_{A\rightarrow B})(\rho_{RA}))\nonumber\\
&  +P((\mathcal{M}_{B\rightarrow C}^{k}\circ\mathcal{N}_{A\rightarrow B}%
)(\rho_{RA}),(\mathcal{M}_{B\rightarrow C}\circ\mathcal{N}_{A\rightarrow
B})(\rho_{RA}))\nonumber\\
&  \leq P(\mathcal{N}_{A\rightarrow B}^{k}(\rho_{RA}),\mathcal{N}%
_{A\rightarrow B}(\rho_{RA}))\nonumber\\
&  +P((\mathcal{M}_{B\rightarrow C}^{k}\circ\mathcal{N}_{A\rightarrow B}%
)(\rho_{RA}),(\mathcal{M}_{B\rightarrow C}\circ\mathcal{N}_{A\rightarrow
B})(\rho_{RA})). \label{eq:serial-comp-str-conv}%
\end{align}
The first inequality follows from the triangle inequality and the second from
data processing. Using the strong convergence of $\{\mathcal{N}_{A\rightarrow
B}^{k}\}_{k}$ and $\{\mathcal{M}_{B\rightarrow C}^{k}\}_{k}$, the fact that
$\mathcal{N}_{A\rightarrow B}(\rho_{RA})$ is a fixed state independent of $k$,
applying the inequality in \eqref{eq:serial-comp-str-conv}, and taking the
limit $k\rightarrow\infty$, we find that%
\begin{multline}
\!\!\!\!\lim_{k\rightarrow\infty}P((\mathcal{M}_{B\rightarrow C}^{k}%
\circ\mathcal{N}_{A\rightarrow B}^{k})(\rho_{RA}),(\mathcal{M}_{B\rightarrow
C}\circ\mathcal{N}_{A\rightarrow B})(\rho_{RA}))\\
=0.
\end{multline}
Since the state $\rho_{RA}$ was arbitrary, the proof is complete.
\end{proof}

\subsection{Strong convergence in the teleportation simulation of bosonic
Gaussian channels}

\label{sec:tele-sim}

The teleportation simulation of a bosonic Gaussian channel is another
important notion to discuss. As found in \cite{NFC09}, single-mode,
phase-covariant bosonic channels, such as the thermal, amplifier, or
additive-noise channels, can be simulated by employing the bosonic
teleportation protocol from \cite{prl1998braunstein}. More general classes of
bosonic Gaussian channels can be simulated as well \cite{WPG07}. In this subsection, I
exclusively discuss single-mode bosonic Gaussian channels and extend the results later to particular multi-mode bosonic Gaussian channels. Denoting the
original channel by $\mathcal{G}$, an unideal teleportation simulation of it
realizes the bosonic Gaussian channel $\mathcal{G}^{\bar{\sigma}}%
\equiv\mathcal{G}\circ\mathcal{T}^{\bar{\sigma}}$, where $\mathcal{T}%
^{\bar{\sigma}}$ is the additive-noise channel from
\eqref{eq:TP-as-add-noise}. This unideal teleportation simulation is possible
due to the displacement covariance of bosonic Gaussian channels. Again, it is
needed to clarify the meaning of the convergence $\mathcal{G}=\lim
_{\bar{\sigma}\rightarrow0}\mathcal{G}^{\bar{\sigma}}$. Based on the previous
discussions in this paper, it is clear that the convergence should be
considered in the strong sense in most applications: for a state $\rho_{RA}$,
we have that%
\begin{equation}
\lim_{\bar{\sigma}\rightarrow0}\left[  1-F((\operatorname{id}_{R}%
\otimes\mathcal{G}_{A})(\rho_{RA}),(\operatorname{id}_{R}\otimes
\mathcal{G}_{A}^{\bar{\sigma}})(\rho_{RA}))\right]  =0,
\end{equation}
where $F$ denotes the quantum fidelity. This equality follows as a consequence
of \eqref{eq:fixed-state-converge} and the data-processing inequality for fidelity.

As a consequence of \eqref{eq:n-converge-tele}\ and data processing, we also
have the following convergence for a teleportation simulation of the
tensor-power channel $\mathcal{G}^{\otimes n}$: for a state $\rho_{RA^{n}}$,
we have that%
\begin{equation}
\lim_{\bar{\sigma}\rightarrow0}\left[  1-F(\mathcal{G}_{A}^{\otimes n}%
(\rho_{RA^{n}}),(\mathcal{G}_{A}^{\bar{\sigma}})^{\otimes n}(\rho_{RA^{n}%
}))\right]  =0,
\end{equation}
where the identity map $\operatorname{id}_{R}$ is omitted for simplicity.

Finally, the argument from Section~\ref{sec:strong-arbitrary}\ applies to the
teleportation simulation of bosonic Gaussian channels as well. In more detail,
the strong convergence of $\mathcal{G}^{\bar{\sigma}}$ to $\mathcal{G}$ in the
limit $\bar{\sigma}\rightarrow0$ implies strong convergence of $n$ uses of
$\mathcal{G}^{\bar{\sigma}}$ to $n$ uses of $\mathcal{G}$ in the general sense
discussed in Section~\ref{sec:strong-arbitrary}. That is, as a consequence of
\eqref{eq:TP-infidelity} and \eqref{eq:strong-conv-arbitrary}, we have that%
\begin{equation}
\sup_{\mathcal{P}^{(n)}}\lim_{\bar{\sigma}\rightarrow0}\varepsilon_{\text{ad}%
}^{(n)}(\bar{\sigma},\mathcal{P}^{(n)})=0,
\label{eq:strong-converge-tele-adap}%
\end{equation}
where $\varepsilon_{\text{ad}}^{(n)}(\bar{\sigma},\mathcal{P}^{(n)})$ is
defined by replacing $\mathcal{N}_{A\rightarrow B}^{k}$ with $\mathcal{G}%
^{\bar{\sigma}}$ and $\mathcal{N}_{A\rightarrow B}$ with $\mathcal{G}$ in
\eqref{eq:omega-k}, \eqref{eq:omega}, and \eqref{eq:adap-err}.


\subsection{Uniform convergence in the teleportation simulations of pure-loss,
thermal, pure-amplifier, amplifier, and additive-noise channels}

\label{sec:unif-conv-loss-amp}I now prove that the teleportation simulations
of pure-loss, thermal, pure-amplifier, amplifier, and additive-noise channels
converge \textit{uniformly} to the original channels, in the limit of ideal
squeezing and detection for the simulations. The argument for uniform
convergence is elementary, using the structure of these channels and their
teleportation simulations, as well as a data processing argument that is the
same as that which was employed in \cite{TW16,SWAT17}. Note that these uniform
convergence results are in contrast to the teleportation simulation of the
ideal channel, where the convergence occurs in the strong sense but
\textit{not} in the uniform sense.

To prove the uniform convergence of the teleportation simulations of the
aforementioned channels, let us start with the thermal channel. Consider that
the thermal channel $\mathcal{L}_{\eta,N_{B}}$\ of transmissivity $\eta
\in(0,1)$ and thermal photon number $N_{B}\geq0$ is completely specified by
its action on the $2\times1$ mean vector $s$ and $2\times2$ covariance matrix
$V$ of a single-mode input \cite{S17}:%
\begin{align}
s  &  \rightarrow Xs,\label{eq:mean-action}\\
V  &  \rightarrow XVX^{T}+Y, \label{eq:cov-action}%
\end{align}
where%
\begin{align}
X  &  =\sqrt{\eta}I_{2},\\
Y  &  =(1-\eta)(2N_{B}+1)I_{2},
\end{align}
and $I_{2}$ denotes the $2\times2$ identity matrix. An unideal teleportation
simulation of a thermal channel is equivalent to the serial concatenation of
the additive-noise channel $\mathcal{T}^{\bar{\sigma}}$\ with variance
$\bar{\sigma}>0$, followed by the thermal channel $\mathcal{L}_{\eta,N_{B}}$,
as discussed in Section~\ref{sec:tele-sim}. Since the additive-noise channel
has the same action as in \eqref{eq:mean-action} and \eqref{eq:cov-action},
but with%
\begin{align}
X  &  =I_{2},\\
Y  &  = 2 \bar{\sigma}I_{2},
\end{align}
we find, after composing and simplifying, that the simulating channel
$\mathcal{L}_{\eta,N_{B}}\circ\mathcal{T}^{\bar{\sigma}}$ has the same action
as in \eqref{eq:mean-action} and \eqref{eq:cov-action}, but with%
\begin{align}
X  &  =\sqrt{\eta}I_{2},\\
Y  &  =\left[  \eta2 \bar{\sigma}+(1-\eta)(2N_{B}+1)\right]  I_{2}\\
&  =(1-\eta)(2\left[  N_{B}+\eta\bar{\sigma}/(1-\eta)\right]  +1)I_{2}.
\end{align}
This latter finding means that the simulating channel $\mathcal{L}_{\eta
,N_{B}}\circ\mathcal{T}^{\bar{\sigma}}$ is equivalent to the thermal channel
$\mathcal{L}_{\eta,N_{B}+\eta\bar{\sigma}/(1-\eta)}$, i.e.,%
\begin{equation}
\mathcal{L}_{\eta,N_{B}}\circ\mathcal{T}^{\bar{\sigma}}=\mathcal{L}%
_{\eta,N_{B}+\eta\bar{\sigma}/(1-\eta)}. \label{eq:tp-sim-therm-is-therm}%
\end{equation}
Let us set%
\begin{equation}
N_{B}^{\prime}\equiv N_{B}+\eta\bar{\sigma}/(1-\eta).
\end{equation}

Note that any thermal channel $\mathcal{L}_{\eta,N_{B}}$ can be realized in
three steps:

\begin{enumerate}
\item prepare an environment mode in a thermal state $\theta(N_{B})$\ of mean
photon number $N_{B}\geq0$, where%
\begin{equation}
\theta(N_{B})=\frac{1}{N_{B}+1}\sum_{n=0}^{\infty}\left(  \frac{N_{B}}%
{N_{B}+1}\right)  ^{n}|n\rangle\langle n|,
\end{equation}

\item interact the channel input mode with the environment mode at a unitary
beamsplitter $\mathcal{B}_{\eta}$ of transmissivity $\eta$, and

\item discard the environment mode.
\end{enumerate}

\noindent This observation and that in \eqref{eq:tp-sim-therm-is-therm}\ are
what lead to uniform convergence of the simulating channel $\mathcal{L}%
_{\eta,N_{B}^{\prime}}$ to the original channel $\mathcal{L}_{\eta,N_{B}}$ in
the limit as $\bar{\sigma}\rightarrow0$. Indeed, let $\rho_{RA}$ be an
arbitrary input state, with $R$ a reference system corresponding to an
arbitrary separable Hilbert space and system $A$ the channel input. Then we
find that%
\begin{align}
&  P((\operatorname{id}_{R}\otimes\mathcal{L}_{\eta,N_{B}})(\rho
_{RA}),(\operatorname{id}_{R}\otimes\mathcal{L}_{\eta,N_{B}^{\prime}}%
)(\rho_{RA}))\nonumber\\
&  \leq P((\operatorname{id}_{R}\otimes\mathcal{B}_{\eta})[\rho_{RA}%
\otimes\theta(N_{B})],(\operatorname{id}_{R}\otimes\mathcal{B}_{\eta}%
)[\rho_{RA}\otimes\theta(N_{B}^{\prime})])\nonumber\\
&  =P(\rho_{RA}\otimes\theta(N_{B}),\rho_{RA}\otimes\theta(N_{B}^{\prime
}))\nonumber\\
&  =P(\theta(N_{B}),\theta(N_{B}^{\prime}))\nonumber\\
&  \equiv e(N_{B},\eta,\bar{\sigma}), \label{eq:uniform-converge-bnd}%
\end{align}
where $e(N_{B},\eta,\bar{\sigma})$ explicitly evaluates to%
\begin{multline}
e(N_{B},\eta,\bar{\sigma})=\label{eq:explicit-eval}\\
\Bigg[1-\Big[\sqrt{\left(  N_{B}+1\right)  \left(  N_{B}+\eta\bar{\sigma
}/(1-\eta)+1\right)  }\\
-\sqrt{N_{B}\left[  N_{B}+\eta\bar{\sigma}/(1-\eta)\right]  }\Big]^{-2}%
\Bigg]^{1/2}.
\end{multline}
The explicit evaluation in \eqref{eq:explicit-eval}\ is a direct consequence
of \cite[Eqs.~(34)--(35)]{TW16}, found by evaluating the fidelity between two
thermal states of respective mean photon numbers $N_{B}$ and $N_{B}+\eta
\bar{\sigma}/(1-\eta)$. See also \cite{H75,S98} for formulas for the fidelity of single-mode Gaussian states. The first inequality in
\eqref{eq:uniform-converge-bnd} follows from data processing. The first
equality follows from unitary invariance of the metric $P$ and the second from
its invariance under tensoring in the same state $\rho_{RA}$. To summarize the
inequality in \eqref{eq:uniform-converge-bnd}, it is stating that the
distinguishability of the channels $\mathcal{L}_{\eta,N_{B}}$ and
$\mathcal{L}_{\eta,N_{B}^{\prime}}$, when allowing for any input probe state
$\rho_{RA}$, is limited by the distinguishability of the environment states
$\theta(N_{B})$ and $\theta(N_{B}^{\prime})$, and this is similar to the
observations made in \cite{TW16,SWAT17}. Thus, the bound in
\eqref{eq:uniform-converge-bnd}\ is a \textit{uniform} bound, holding for all
input states $\rho_{RA}$, and so we conclude that%
\begin{multline}
\sup_{\rho_{RA}}P((\operatorname{id}_{R}\otimes\mathcal{L}_{\eta,N_{B}}%
)(\rho_{RA}),(\operatorname{id}_{R}\otimes\mathcal{L}_{\eta,N_{B}^{\prime}%
})(\rho_{RA}))\\
\leq e(N_{B},\eta,\bar{\sigma}). \label{eq:loss-unif-bnd}%
\end{multline}
Now taking the limit $\bar{\sigma}\rightarrow0$ and using the fact that
$\lim_{\bar{\sigma}\rightarrow0}e(N_{B},\eta,\bar{\sigma})=0$ for $\eta
\in(0,1)$ and $N_{B}\geq0$, we find that%
\begin{equation}
\lim_{\bar{\sigma}\rightarrow0}\sup_{\rho_{RA}}P((\operatorname{id}_{R}%
\otimes\mathcal{L}_{\eta,N_{B}})(\rho_{RA}),(\operatorname{id}_{R}%
\otimes\mathcal{L}_{\eta,N_{B}^{\prime}})(\rho_{RA}))=0.
\end{equation}
Thus, the teleportation simulation $\mathcal{L}_{\eta,N_{B}}\circ
\mathcal{T}^{\bar{\sigma}}$ of the thermal channel $\mathcal{L}_{\eta,N_{B}}$
of transmissivity $\eta\in(0,1)$ and thermal photon number $N_{B}\geq0$
converges \textit{uniformly} to the thermal channel.

The above uniform convergence result holds in particular for a pure-loss
channel of transmissivity $\eta\in(0,1)$, because this channel is a thermal
channel with $N_{B}=0$. That is, the environment state for the pure-loss
channel is a vacuum state, and its teleportation simulation is a thermal
channel with the same transmissivity and environment state given by a thermal
state of mean photon number $\eta\bar{\sigma}/(1-\eta)$. In this case, the
uniform upper bound $e(N_{B}=0,\eta,\bar{\sigma})$ simplifies to%
\begin{equation}
e(N_{B}=0,\eta,\bar{\sigma})=\sqrt{1-\frac{1}{\eta\bar{\sigma}/(1-\eta)+1}},
\end{equation}
for which we clearly have that $\lim_{\bar{\sigma}\rightarrow0}e(N_{B}%
=0,\eta,\bar{\sigma})=0$. Thus, the teleportation simulation of a pure-loss
channel $\mathcal{L}_{\eta,N_{B}=0}$ converges \textit{uniformly} to
$\mathcal{L}_{\eta,N_{B}=0}$.

Similar results hold for the pure-amplifier and amplifier channels. Indeed, to
see this, let us begin by considering a general amplifier channel
$\mathcal{A}_{G,N_{B}}$\ of gain $G>1$ and thermal photon number $N_{B}\geq0$.
Such a channel has the action as in \eqref{eq:mean-action} and
\eqref{eq:cov-action}, but with%
\begin{align}
X  &  =\sqrt{G}I_{2},\\
Y  &  =(G-1)(2N_{B}+1)I_{2}.
\end{align}
By similar reasoning as before, the teleportation simulation $\mathcal{A}%
_{G,N_{B}}\circ\mathcal{T}^{\bar{\sigma}}$ of the amplifier channel has the
action as in \eqref{eq:mean-action} and \eqref{eq:cov-action}, but with%
\begin{align}
X  &  =\sqrt{G}I_{2},\\
Y  &  =\left[  G 2 \bar{\sigma}+(G-1)(2N_{B}+1)\right]  I_{2}\\
&  =(G-1)(2\left[  N_{B}+G\bar{\sigma}/(G-1)\right]  +1)I_{2}.
\end{align}
Thus, the teleportation simulation $\mathcal{A}_{G,N_{B}}\circ\mathcal{T}%
^{\bar{\sigma}}$\ is equivalent to an amplifier channel $\mathcal{A}%
_{G,N_{B}+G\bar{\sigma}/(G-1)}$:%
\begin{equation}
\mathcal{A}_{G,N_{B}}\circ\mathcal{T}^{\bar{\sigma}}=\mathcal{A}%
_{G,N_{B}+G\bar{\sigma}/(G-1)}.
\end{equation}
An amplifier channel $\mathcal{A}_{G,N_{B}}$ can be realized by the following
three steps:

\begin{enumerate}
\item prepare an environment mode in a thermal state $\theta(N_{B})$\ of mean
photon number $N_{B}\geq0$,

\item interact the channel input mode with the environment mode using a
unitary two-mode squeezer $\mathcal{S}_{G}$ of gain $G$, and

\item discard the environment mode.
\end{enumerate}

For an arbitrary state $\rho_{RA}$, we find the following upper bound by the
same reasoning as in \eqref{eq:uniform-converge-bnd}, but replacing the
beamsplitter $\mathcal{B}_{\eta}$\ therein by the two-mode squeezer
$\mathcal{S}_{G}$,%
\begin{multline}
P((\operatorname{id}_{R}\otimes\mathcal{A}_{G,N_{B}})(\rho_{RA}%
),(\operatorname{id}_{R}\otimes\mathcal{A}_{G,N_{B}^{\prime\prime}})(\rho
_{RA}))\label{eq:amp-unif-bnd}\\
\leq P(\theta(N_{B}),\theta(N_{B}^{\prime\prime}))\equiv e(N_{B},G,\bar
{\sigma})
\end{multline}
where%
\begin{equation}
N_{B}^{\prime\prime}\equiv N_{B}+G\bar{\sigma}/(G-1)
\end{equation}
and%
\begin{multline}
e(N_{B},G,\bar{\sigma})=\\
\Bigg[1-\Big[\sqrt{\left(  N_{B}+1\right)  \left(  N_{B}+G\bar{\sigma
}/(G-1)+1\right)  }\\
-\sqrt{N_{B}\left[  N_{B}+G\bar{\sigma}/(G-1)\right]  }\Big]^{-2}\Bigg]^{1/2}.
\end{multline}
Again, the inequality in \eqref{eq:amp-unif-bnd}\ is the statement that the
distinguishability of the channels $\mathcal{A}_{G,N_{B}}$ and $\mathcal{A}%
_{G,N_{B}^{\prime\prime}}$, when allowing for any input probe state $\rho
_{RA}$, is limited by the distinguishability of the channel environment states
$\theta(N_{B})$ and $\theta(N_{B}^{\prime\prime})$. Given that the bound in
\eqref{eq:amp-unif-bnd} is a \textit{uniform} bound holding for all input
states $\rho_{RA}$, this implies that%
\begin{multline}
\sup_{\rho_{RA}}P((\operatorname{id}_{R}\otimes\mathcal{A}_{G,N_{B}}%
)(\rho_{RA}),(\operatorname{id}_{R}\otimes\mathcal{A}_{G,N_{B}^{\prime\prime}%
})(\rho_{RA}))\\
\leq e(N_{B},G,\bar{\sigma}) \label{eq:amp-unif-bnd-final}%
\end{multline}
We can then take the limit $\bar{\sigma}\rightarrow0$ and use the fact that
$\lim_{\bar{\sigma}\rightarrow0}e(N_{B},G,\bar{\sigma})=0$ for all $G>1$ and
$N_{B}\geq0$ to find that%
\begin{equation}
\lim_{\bar{\sigma}\rightarrow0}\sup_{\rho_{RA}}P((\operatorname{id}_{R}%
\otimes\mathcal{A}_{G,N_{B}})(\rho_{RA}),(\operatorname{id}_{R}\otimes
\mathcal{A}_{G,N_{B}^{\prime\prime}})(\rho_{RA}))=0.
\end{equation}
Thus, the teleportation simulation $\mathcal{A}_{G,N_{B}}\circ\mathcal{T}%
^{\bar{\sigma}}$ of the amplifier channel $\mathcal{A}_{G,N_{B}}$\ converges
\textit{uniformly} to it, for all $G>1$ and thermal photon number $N_{B}\geq0$.

The pure-amplifier channel is a special case of the amplifier channel
$\mathcal{A}_{G,N_{B}}$ with $N_{B}=0$, so that the above analysis applies and
the teleportation simulation of the pure-amplifier channel $\mathcal{A}%
_{G,N_{B}=0}$\ converges \textit{uniformly} to it. Indeed, the uniform upper
bound $e(N_{B}=0,G,\bar{\sigma})$ simplifies as%
\begin{multline}
P((\operatorname{id}_{R}\otimes\mathcal{A}_{G,N_{B}=0})(\rho_{RA}%
),(\operatorname{id}_{R}\otimes\mathcal{A}_{G,G\bar{\sigma}/(G-1)})(\rho
_{RA}))\\
\leq e(N_{B}=0,G,\bar{\sigma})=\sqrt{1-\frac{1}{G\bar{\sigma}/(G-1)+1}},
\end{multline}
and so it is clear that%
\begin{equation}
\lim_{\bar{\sigma}\rightarrow0}\sup_{\rho_{RA}}P(\mathcal{A}_{G,0}(\rho
_{RA}),\mathcal{A}_{G,G\bar{\sigma}/(G-1)}(\rho_{RA}))=0,
\end{equation}
where I have omitted the identity maps $\operatorname{id}_{R}$ acting on
system $R$ for simplicity.

A similar argument establishes that the teleportation simulation of the
additive-noise channel $\mathcal{T}^{\xi}$\ with variance $\xi>0$ converges
uniformly to it. To see this, let us begin by noting that any additive-noise
channel $\mathcal{T}^{\xi}$ can be realized by the following three steps:

\begin{enumerate}
\item Prepare a continuous classical environment register according to the
complex Gaussian distribution $G_{\xi}(\alpha)$, as defined in \eqref{eq:complex-Gauss}.

\item Based on the classical value $\alpha$ in the environment register, apply
a unitary displacement operation $D(\alpha)$ to the channel input. This step
can be described as an interaction channel $\mathcal{C}$ between the channel
input and the environment register.

\item Finally, discard the environment register.
\end{enumerate}

Also, note that the fidelity between two complex Gaussian distributions of
variances $\xi_{1},\xi_{2}>0$ is given by%
\begin{equation}
F(G_{\xi_{1}},G_{\xi_{2}})=\frac{4\xi_{1}\xi_{2}}{\left(  \xi_{1}+\xi
_{2}\right)  ^{2}}, \label{eq:fid-classical-Gauss}%
\end{equation}
which one may verify directly or consult \cite{CA79}. Then, proceeding as in
the previous proofs, we have for an arbitrary input state $\rho_{RA}$ that%
\begin{align}
&  P((\operatorname{id}_{R}\otimes\mathcal{T}^{\xi})(\rho_{RA}%
),(\operatorname{id}_{R}\otimes\mathcal{T}^{\xi+\bar{\sigma}})(\rho
_{RA}))\nonumber\\
&  \leq P((\operatorname{id}_{R}\otimes\mathcal{C})(\rho_{RA}\otimes G_{\xi
}),(\operatorname{id}_{R}\otimes\mathcal{C})(\rho_{RA}\otimes G_{\xi
+\bar{\sigma}}))\nonumber\\
&  \leq P(\rho_{RA}\otimes G_{\xi},\rho_{RA}\otimes G_{\xi+\bar{\sigma}%
})\nonumber\\
&  =P(G_{\xi},G_{\xi+\bar{\sigma}})\nonumber\\
&  =\sqrt{1-\frac{4\xi(\xi+\bar{\sigma})}{(2\xi+\bar{\sigma})^{2}}}.
\label{ea:add-noise-steps-unif-bnd}%
\end{align}
The first and second inequalities follow from data processing. The first
equality follows because the metric $P$ is invariant under tensoring in the
same state $\rho_{RA}$. The final equality follows from the definition of the
metric $P$ and the formula in \eqref{eq:fid-classical-Gauss}. As in the other
proofs, the inequality in \eqref{ea:add-noise-steps-unif-bnd} is intuitive,
indicating that the distinguishability of the channels $\mathcal{T}^{\xi}$ and
$\mathcal{T}^{\xi+\bar{\sigma}}$ is limited by the distinguishability of the
underlying classical distributions $G_{\xi}$ and $G_{\xi+\bar{\sigma}}$. The
bound in \eqref{ea:add-noise-steps-unif-bnd}\ is a \textit{uniform} bound,
holding for all input states $\rho_{RA}$, and so we conclude that%
\begin{multline}
\sup_{\rho_{RA}}P((\operatorname{id}_{R}\otimes\mathcal{T}^{\xi})(\rho
_{RA}),(\operatorname{id}_{R}\otimes\mathcal{T}^{\xi+\bar{\sigma}})(\rho
_{RA}))\\
\leq\sqrt{1-\frac{4\xi(\xi+\bar{\sigma})}{(2\xi+\bar{\sigma})^{2}}}.
\label{eq:add-noise-unif-upp-bnd}%
\end{multline}
Finally, we take the limit as $\bar{\sigma}\rightarrow0$ to establish the
uniform convergence of the teleportation simulation of the additive-noise
channel of variance $\xi>0$ to itself:%
\begin{equation}
\lim_{\bar{\sigma}\rightarrow0}\sup_{\rho_{RA}}P((\operatorname{id}_{R}%
\otimes\mathcal{T}^{\xi})(\rho_{RA}),(\operatorname{id}_{R}\otimes
\mathcal{T}^{\xi+\bar{\sigma}})(\rho_{RA}))=0.
\end{equation}


\begin{remark}
By the results of \cite{TW16}, the uniform upper bounds in
\eqref{eq:loss-unif-bnd}, \eqref{eq:amp-unif-bnd-final}, and
\eqref{eq:add-noise-unif-upp-bnd} are all achievable and are thus equalities.
The achievable strategy consists of taking the input states $\rho_{RA}$ to be
a sequence of two-mode squeezed vacuum states with photon number $N_{S}$
tending to infinity.
\end{remark}

\begin{remark}
On the one hand, the thermal, amplifier, and additive-noise channels are the single-mode
bosonic Gaussian channels that are of major interest in applications, as
stressed in \cite[Section~3.5]{HG12} and \cite[Section~12.6.3]{H12}. On the
other hand, one could consider generalizing the results of this section to
arbitrary single-mode bosonic Gaussian channels.
In doing so, one should consider the Holevo classification of single-mode
bosonic Gaussian channels \cite{Holevo2007}. However, there is little reason
to generalize the contents of this section to other channels in the Holevo
classification. The thermal and amplifier channels form the class $C$
discussed in \cite{Holevo2007}, and the additive-noise channels form the class
$B_{2}$ from \cite{Holevo2007}, which I have already considered in this
section. The classes that remain are labeled $A$, $B_{1}$, and $D$. The
channels in classes $A$ and $D$ are \textit{entanglement breaking}, as proved
in \cite{Holevo2008}. Thus, the channels in classes $A$ and $D$ can be
directly realized by the action of an LOCC on the input state (without any
need for an entangled resource state), and thus we would never have any reason
to be interested in the teleportation simulation of these channels. The final
remaining class is $B_{1}$,
but channels in the class $B_{1}$ do not seem to be interesting in physical applications.
\end{remark}

\begin{remark}
The class $B_{1}$ has been considered in \cite{P18}, where it was
shown that the teleportation simulation of a channel in this class does not
converge uniformly, similar to what occurs for the identity channel. Based on
the above, and the fact that the ideal channel and channels in the class
$B_{1}$ have unconstrained quantum capacity equal to infinity
\cite{Holevo2007}, as well as the fact that channels in the classes $C$ and
$B_{2}$ have finite unconstrained quantum capacity \cite{SWAT17}, we can
conclude that, among the single-mode bosonic Gaussian channels that are not
entanglement breaking, their teleportation simulations converge uniformly if
and only if their unconstrained quantum capacity is finite. This establishes a
non-trivial link between teleportation simulation and unconstrained quantum capacity.
\end{remark}

\subsection{Generalization of uniform convergence results to multi-mode bosonic
Gaussian channels}

In this section, I discuss a generalization of the results of the previous
section to the case of multi-mode bosonic Gaussian channels \footnote{I am
indebted to an anonymous referee for suggesting this multi-mode
generalization.}. Before doing so, I give a brief review of bosonic Gaussian
states and channels (see \cite{CEGH08} for a more comprehensive review). Let%
\begin{equation}
\hat{R}\equiv\left[  \hat{q}_{1},\ldots,\hat{q}_{m},\hat{p}_{1},\ldots,\hat
{p}_{m}\right]  \equiv\left[  \hat{x}_{1},\ldots,\hat{x}_{2m}\right]
\end{equation}
denote a row vector of position- and momentum-quadrature operators, satisfying
the canonical commutation relations:%
\begin{equation}
\left[  \hat{R}_{j},\hat{R}_{k}\right]  =i\Omega_{j,k},\quad\text{where}%
\quad\Omega\equiv%
\begin{bmatrix}
0 & 1\\
-1 & 0
\end{bmatrix}
\otimes I_{m},
\end{equation}
and $I_{m}$ denotes the $m\times m$ identity matrix. We take the annihilation
operator for the $j$th mode as $\hat{a}_{j}=(\hat{q}_{j}+i\hat{p}_{j}%
)/\sqrt{2}$. For $z$ a column vector in $\mathbb{R}^{2m}$, we define the
unitary displacement operator $D(z)=D^{\dagger}(-z)\equiv\exp(i\hat{R}z)$.
Displacement operators satisfy the following relation:%
\begin{equation}
D(z)D(z^{\prime})=D(z+z^{\prime})\exp\!\left(  -\frac{i}{2}z^{T}\Omega
z^{\prime}\right)  .
\end{equation}
Every state $\rho\in\mathcal{D}(\mathcal{H})$ has a corresponding Wigner
characteristic function, defined as%
\begin{equation}
\chi_{\rho}(z)\equiv\operatorname{Tr}\{D(z)\rho\},
\end{equation}
and from which we can obtain the state $\rho$ as%
\begin{equation}
\rho=\int\frac{d^{2m}z}{\left(  2\pi\right)  ^{m}}\ \chi_{\rho}(z)\ D^{\dag
}(z).
\end{equation}
A quantum state $\rho$ is Gaussian if its Wigner characteristic function has a
Gaussian form as%
\begin{equation}
\chi_{\rho}(\xi)=\exp\left(  -\frac{1}{4}z^{T}V^{\rho}z+i\left[  \mu^{\rho
}\right]  ^{T}z\right)  ,
\end{equation}
where $\mu^{\rho}$ is the $2m\times1$ mean vector of $\rho$, whose entries are
defined by $\mu_{j}^{\rho}\equiv\langle\hat{R}_{j}\rangle_{\rho}$ and
$V^{\rho}$ is the $2m\times2m$ covariance matrix of $\rho$, whose entries are
defined as%
\begin{equation}
V_{j,k}^{\rho}\equiv\langle\{\hat{R}_{j}-\mu_{j}^{\rho},\hat{R}_{k}-\mu
_{k}^{\rho}\}\rangle_{\rho}.
\end{equation}
The following condition holds for a valid covariance matrix: $V\geq i\Omega$,
which is a manifestation of the uncertainty principle.

A $2m\times2m$ matrix $S$ is symplectic if it preserves the symplectic form:
$S\Omega S^{T}=\Omega$. According to Williamson's theorem \cite{W36}, there is
a diagonalization of the covariance matrix $V^{\rho}$ of the form,
\begin{equation}
V^{\rho}=S^{\rho}\left(  D^{\rho}\oplus D^{\rho}\right)  \left(  S^{\rho
}\right)  ^{T},
\end{equation}
where $S^{\rho}$ is a symplectic matrix and $D^{\rho}\equiv\operatorname{diag}%
(\nu_{1},\ldots,\nu_{m})$ is a diagonal matrix of symplectic eigenvalues such
that $\nu_{i}\geq1$ for all $i\in\left\{  1,\ldots,m\right\}  $. Computing
this decomposition is equivalent to diagonalizing the matrix $iV^{\rho}\Omega$
\cite[Appendix~A]{WTLB16}.

The Hilbert--Schmidt adjoint of a Gaussian quantum channel $\mathcal{N}_{X,Y}%
$\ from $m$ modes to $m$ modes has the following effect on a displacement
operator $D(z)$ \cite{CEGH08}:%
\begin{equation}
D(z)\longmapsto D(Xz)\exp\left(  -\frac{1}{4}z^{T}Yz+iz^{T}d\right)
,\label{eq:G-chan-1}%
\end{equation}
where $X$ is a real $2m\times2m$ matrix, $Y$ is a real $2m\times2m$ positive
semi-definite matrix, and $d\in\mathbb{R}^{2m}$, such that they satisfy%
\begin{equation}
Y\geq i\left[  \Omega-X^{T}\Omega X\right]  .\label{eq:Gaussian-CP-condition}%
\end{equation}
The effect of the channel on the mean vector $\mu^{\rho}$ and the covariance
matrix $V^{\rho}$\ is thus as follows:%
\begin{align}
\mu^{\rho} &  \longmapsto X^{T}\mu^{\rho}+d,\\
V^{\rho} &  \longmapsto X^{T}V^{\rho}X+Y.\label{eq:G-chan-3}%
\end{align}
All Gaussian channels are covariant with respect to displacement operators,
and this is the main reason why they are teleportation simulable, as noted in
\cite{WPG07,NFC09}. That is, the following relation holds%
\begin{equation}
\mathcal{N}_{X,Y}(D(z)\rho D^{\dag}(z))=D(X^{T}z)\mathcal{N}_{X,Y}%
(\rho)D^{\dag}(X^{T}z).\label{eq:covariance-gaussian}%
\end{equation}

Just as every quantum channel can be implemented as a unitary transformation
on a larger space followed by a partial trace \cite{S55}, so can Gaussian
channels be implemented as a Gaussian unitary on a larger space with some
extra modes prepared in the vacuum state, followed by a partial trace
\cite{CEGH08}. Given a Gaussian channel $\mathcal{N}_{X,Y}$\ with $Z$ such
that $Y=ZZ^{T}$ we can find two other matrices $X_{E}$ and $Z_{E}$ such that
there is a symplectic matrix
\begin{equation}
S=%
\begin{bmatrix}
X^{T} & Z\\
X_{E}^{T} & Z_{E}%
\end{bmatrix}
,\label{eq:gaussian-dilation}%
\end{equation}
which corresponds to the Gaussian unitary transformation on a larger space.

Alternatively, for certain Gaussian channels, there is a realization that is
analogous to those discussed in the previous section for the thermal,
amplifier, and additive-noise channels. In particular, consider a Gaussian
channel with $X$ and $Y$ matrices as discussed above, and suppose that the
mean vector $d=0$ (note that the condition $d=0$ is not particularly
restrictive because it just corresponds to a unitary displacement at the input
or output of the channel, and capacities or distinguishability measures do not
change under such unitary actions). Whenever the matrix $\Omega-X^{T}\Omega X$
is full rank, implying then by \eqref{eq:Gaussian-CP-condition} that $Y$ is
full rank, the Gaussian channel can be realized as follows \cite[Theorem~1]%
{CEGH08}:%
\begin{equation}
\rho_{A}\rightarrow\mathcal{N}_{X,Y}(\rho_{A})=\operatorname{Tr}_{E}%
\{U_{AE}(\rho_{A}\otimes\gamma_{E}(Y))U_{AE}^{\dag}\}, \label{eq:multi-mode-env}
\end{equation}
where $\rho_{A}$ is the $m$-mode input, $\gamma_{E}(Y)$ is a zero-mean, $m$-mode
Gaussian state with covariance matrix given by $KYK^{T}$, where $K$ is the
invertible matrix discussed around \cite[Eqs.~(27)--(28)]{CEGH08}, and
$U_{AE}$ is a Gaussian unitary acting on $2m$ modes.

Given the above result, we can then generalize the argument from the previous
section to argue that the teleportation simulations of these channels converge
uniformly. Indeed, as discussed in \cite{WPG07} and as a generalization of the
single-mode case discussed previously, the teleportation simulation of a
Gaussian channel $\mathcal{N}_{X,Y}$\ realizes the Gaussian channel
$\mathcal{N}_{X,Y+\bar{\sigma}I}$, where $\bar{\sigma}>0$ is a parameter
characterizing the squeezing strength and the unideal detections involved in
the teleportation simulation. Thus, as before, the teleportation simulation of
a Gaussian channel simply acts as an additive-noise channel concatenated with
the original channel, and the effect is that the noise matrix for the channel
realized from the teleportation simulation is $Y+\bar{\sigma}I$, while the $X$ matrix is unaffected. Thus,
invoking \cite[Theorem~1]{CEGH08}, the teleportation simulation $\mathcal{N}%
_{X,Y+\bar{\sigma}I}$ can be realized as the following transformation:%
\begin{equation}
\rho_{A}\rightarrow\mathcal{N}_{X,Y}(\rho_{A})=\operatorname{Tr}_{E}%
\{U_{AE}(\rho_{A}\otimes\gamma_{E}(Y+\bar{\sigma}I))U_{AE}^{\dag}\}.
\end{equation}
Then we are led to the following theorem:

\begin{theorem}
Let $\mathcal{N}_{X,Y}$ be a multi-mode quantum Gaussian channel of the form in
\eqref{eq:G-chan-1}--\eqref{eq:G-chan-3}, such that $\Omega-X^{T}\Omega X$ is
full rank. Then its teleportation simulation converges uniformly, in the sense
that%
\begin{multline}
\sup_{\rho_{RA}}P[(\operatorname{id}_{R}\otimes\mathcal{N}_{X,Y})(\rho
_{RA}),(\operatorname{id}_{R}\otimes\mathcal{N}_{X,Y+\bar{\sigma}I})(\rho
_{RA})]\label{eq:multimode-unif-bnd}\\
\leq P(\gamma_{E}(Y),\gamma_{E}(Y+\bar{\sigma}I)),
\end{multline}
where $\gamma_{E}(Y)$ is defined in \eqref{eq:multi-mode-env},
\begin{equation}
\lim_{\bar{\sigma}\rightarrow0}P(\gamma_{E}(Y),\gamma_{E}(Y+\bar
{\sigma}I))=0,
\end{equation}
and one can use the explicit formula \cite[Section~IV]%
{PS00}\ for the fidelity of multi-mode, zero-mean Gaussian states to find an analytical expression for
\begin{equation}
P(\gamma_{E}(Y),\gamma_{E}(Y+\bar{\sigma}I))=\sqrt{1-F(\gamma_{E}%
(Y),\gamma_{E}(Y+\bar{\sigma}I))},
\end{equation}
for $\bar{\sigma}>0$.
\end{theorem}

\begin{proof}
The proof of this theorem follows the same strategy given in the previous
section, which in turn was used in \cite{TW16,SWAT17}. In detail, letting $\rho_{RA}$ be an arbitrary state, we have that%
\begin{align}
& P[(\operatorname{id}_{R}\otimes\mathcal{N}_{X,Y})(\rho_{RA}%
),(\operatorname{id}_{R}\otimes\mathcal{N}_{X,Y+\bar{\sigma}I})(\rho
_{RA})]\nonumber\\
& \leq P[U_{AE}(\rho_{RA}\otimes\gamma_{E}(Y))U_{AE}^{\dag},U_{AE}(\rho
_{RA}\otimes\gamma_{E}(Y^{\bar{\sigma}}))U_{AE}^{\dag}]\nonumber\\
& =P[\rho_{RA}\otimes\gamma_{E}(Y),\rho_{RA}\otimes\gamma_{E}(Y+\bar{\sigma
}I)]\nonumber\\
& =P[\gamma_{E}(Y),\gamma_{E}(Y+\bar{\sigma}I)],
\end{align}
where $Y^{\bar{\sigma}}\equiv Y+\bar{\sigma}I$. The justification of these
steps are the same as before, namely, data processing and unitary invariance.
The above bound is clearly a uniform bound, holding for all states $\rho_{RA}%
$, and so we conclude \eqref{eq:multimode-unif-bnd}.
\end{proof}

\section{Physical Interpretation with the CV Teleportation Game}

\label{sec:TP-game}

We can interpret the results in the previous sections of this paper in a
game-theoretic way, in order to further elucidate the physical meaning of the
two kinds of convergence that we have considered in this paper. Let us
consider a competitive game (call it the CV\ Teleportation Game) between a
Distinguisher and a Teleporter, while an independent Referee determines who
wins the game. At the outset, the Referee flips an unbiased coin and tells the
outcome to the Teleporter. If the coin outcome is heads, then the Teleporter
will apply the ideal channel. If it is tails, then he will apply the
CV\ teleportation protocol. The game is such that either the Distinguisher
reveals his strategy to the Teleporter beforehand, or vice versa. Furthermore,
the Referee always learns the strategies of both the Distinguisher and the
Teleporter. Everyone involved plays honestly. A particular instance of the
game is depicted in Figure~\ref{fig:CV-tele-game}.

\begin{figure}[ptb]
\begin{center}
\includegraphics[
width=3.5in
]{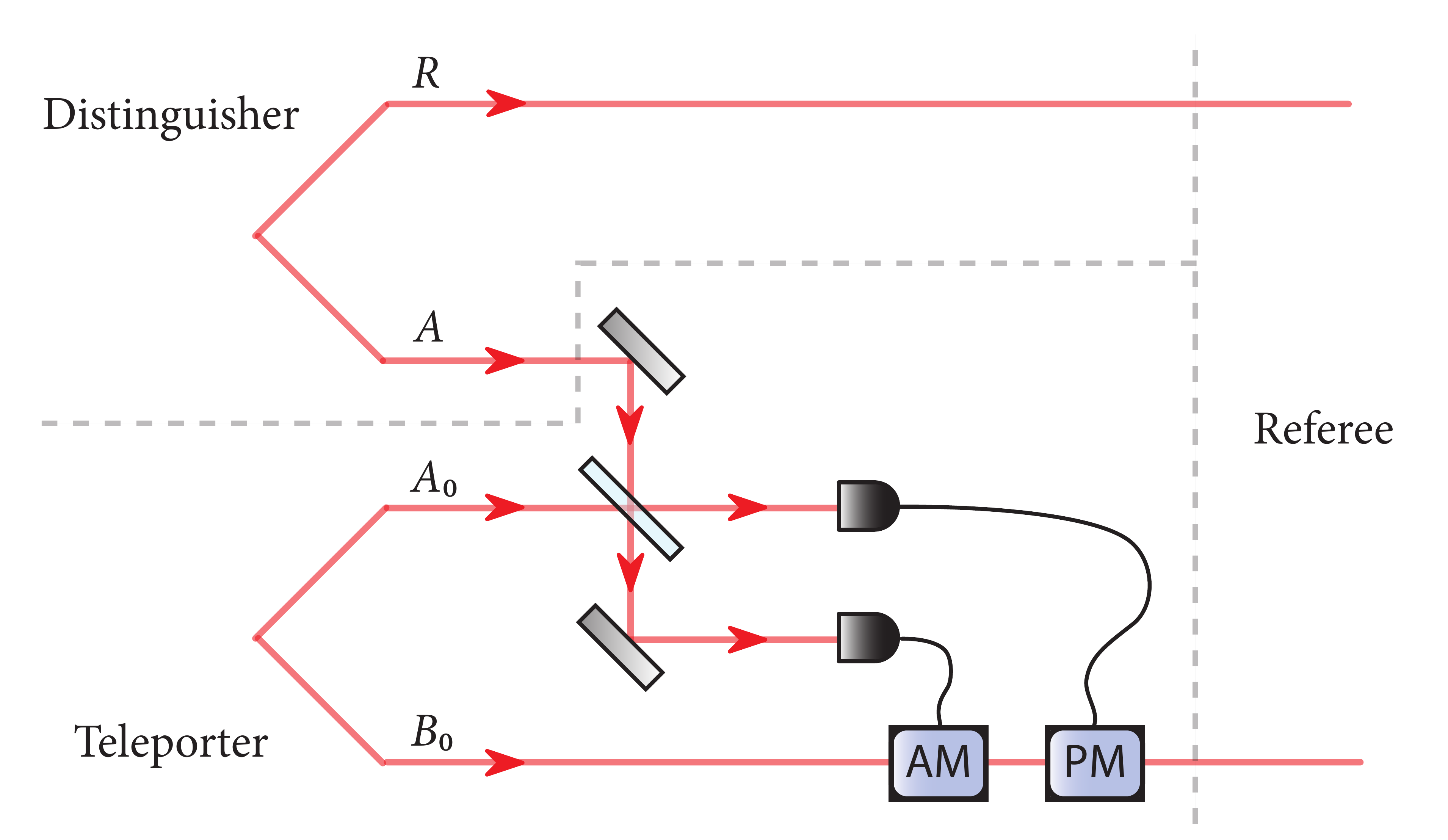}
\end{center}
\caption{Depiction of the CV Teleportation Game, in the case that the coin
outcome is tails, so that the Teleporter applies the continuous-variable
bosonic teleportation protocol to the mode $A$ that the Distinguisher sends.}%
\label{fig:CV-tele-game}%
\end{figure}

I now outline the full game in the case that the Distinguisher reveals his
strategy to the Teleporter. In this case, the Distinguisher picks a pure state
$\psi_{RA}$ and sends mode $A$ to the Teleporter and mode $R$ to the Referee.
The Distinguisher also reveals a classical description of the state $\psi
_{RA}$ to the Referee, and also to the Teleporter in this case. Based on this,
the Teleporter can compute the entanglement infidelity $\varepsilon
(\bar{\sigma},\psi_{RA})$ from \eqref{eq:ent-infidel} and can adjust the
teleportation imperfection $\bar{\sigma}>0$ of his setup accordingly such that
$\varepsilon(\bar{\sigma},\psi_{RA})\approx0 $. The Referee then reports the
coin flip outcome to the Teleporter. If heads, then the Teleporter does
nothing to mode $A$ (ideal channel); if tails, the Teleporter applies the
continuous-variable bosonic teleportation protocol to mode $A$. The Teleporter
sends the output mode to the Referee. The Referee then performs the optimal
binary measurement \cite{H69,H73hel,Hel76} to distinguish the two possible
resulting states. If the measurement outcome is \textquotedblleft
heads,\textquotedblright\ then the channel applied by the Teleporter is
decided to be the ideal channel. If the measurement outcome is
\textquotedblleft tails,\textquotedblright\ then the channel applied by the
Teleporter is decided to be the teleportation channel. This procedure is then
repeated a large number of times. If the fraction of rounds in which the coin
flips match exceeds $3/4$, then the Distinguisher wins. Otherwise, the
Teleporter wins.

To analyze the above physical setup, consider that the probability of
distinguishing the channels in any single round is given by
\cite{H69,H73hel,Hel76}%
\begin{equation}
\Pr\{X=Y\}=\frac{1}{2}\left(  1+\frac{1}{2}\left\Vert \psi_{RA}-\mathcal{T}%
_{A}^{\bar{\sigma}}(\psi_{RA})\right\Vert _{1}\right)  ,
\end{equation}
where $X$ is a Bernoulli random variable modeling the coin flip and $Y$ is a
Bernoulli random variable modeling the measurement outcome. In the above
described case in which the Distinguisher reveals his strategy, this means
that the Teleporter can follow and choose $\bar{\sigma}$ as small as needed to
guarantee that $\Pr\{X=Y\}<3/4$. Thus, with this structure to the game, the
Teleporter wins with high probability after a large number of repetitions. In
fact, based on well known relations between trace distance and fidelity
\cite{Kholevo1972,FG98}, and the result given in \eqref{eq:TP-infidelity}, we
have that%
\begin{equation}
\sup_{\psi_{RA}}\inf_{\bar{\sigma}>0}\left\Vert \psi_{RA}-\mathcal{T}%
_{A}^{\bar{\sigma}}(\psi_{RA})\right\Vert _{1}=0.
\end{equation}

Now suppose that the opposite scenario occurs in which the Teleporter reveals
his strategy (and commits to it). This means that the teleportation
imperfection $\bar{\sigma}>0$ is fixed at the outset. Then the Distinguisher
can choose his input state to be the two-mode squeezed vacuum state
$\Phi(N_{S})_{RA}$ such that $\varepsilon(\bar{\sigma},\Phi(N_{S}%
)_{RA})\approx1$, as considered in \eqref{eq:bad-converge}. This in turn means
that the Distinguisher can guarantee that $\Pr\{X=Y\}>3/4$. Thus, in this
case, the Distinguisher wins with high probability after a large number of
repetitions. In fact, in this latter case, as a consequence of
\eqref{eq:bad-converge}, we have that%
\begin{equation}
\inf_{\bar{\sigma}>0}\sup_{\psi_{RA}}\left\Vert \psi_{RA}-\mathcal{T}%
_{A}^{\bar{\sigma}}(\psi_{RA})\right\Vert _{1}=2.
\end{equation}

One may criticize whether the above game is truly physical. Indeed, it is
never possible in practice to apply the ideal channel. Depending on the
physical situation, the actual channel might be a pure-loss, thermal,
pure-amplifier, amplifier or additive-noise channel, for example (one could
further criticize ``pure-loss'' or ``pure-amplifier'', but let us leave that).
Let $\mathcal{G}_{A}$ denote one of these single-mode, phase-insensitive
bosonic Gaussian channels. Suppose instead that the game changes in the
following way: If the coin outcome is heads, then the Teleporter will apply
the channel $\mathcal{G}$. If it is tails, then he will apply the
teleportation simulation $\mathcal{G}_{A}\mathcal{\circ T}_{A}^{\bar{\sigma}%
}$\ of $\mathcal{G}_{A}$.

There is a striking, physically observable difference in this case. No matter
whether the Distinguisher reveals his strategy to the Teleporter, or the other
way around, the Teleporter always wins with high probability! This is a direct
consequence of the inequalities in \eqref{eq:loss-unif-bnd},
\eqref{eq:amp-unif-bnd-final}, and \eqref{eq:add-noise-unif-upp-bnd}. Indeed,
independent of the revealing, the Teleporter can simply compute the
teleportation imperfection $\bar{\sigma}>0$, while incorporating his knowledge
of the channel parameters, in order to always guarantee that $\Pr\{X=Y\}<3/4$.
Thus, as a consequence of the mathematical fact that the teleportation
simulations of these Gaussian channels converge both uniformly and strongly,
so that%
\begin{multline}
\inf_{\bar{\sigma}>0}\sup_{\psi_{RA}}\left\Vert \mathcal{G}_{A}\mathcal{(}%
\psi_{RA})-\mathcal{G}_{A}\mathcal{\circ T}_{A}^{\bar{\sigma}}(\psi
_{RA})\right\Vert _{1}\\
=\sup_{\psi_{RA}}\inf_{\bar{\sigma}>0}\left\Vert \mathcal{G}_{A}%
\mathcal{(}\psi_{RA})-\mathcal{G}_{A}\mathcal{\circ T}_{A}^{\bar{\sigma}}%
(\psi_{RA})\right\Vert _{1}=0,
\end{multline}
the \textit{physically observable consequence} is that the Teleporter always
has the advantage in this modified (and physically more realistic)\ version of
the CV\ Teleportation Game.

Note that other variations of the CV Teleportation Game are possible. One
could allow for the Distinguisher to employ entangled strategies among the
different rounds, or even adaptive channels in between rounds. The results
given in Propositions~\ref{prop:parallel-comp} and \ref{prop:sequential-comp}
can be used to analyze these other, richer variations of the CV Teleportation
Game, but here I will not go into the details.

\section{Non-asymptotic secret-key-agreement capacities}

\label{sec:non-asymp-SKC}I now briefly review secret-key-agreement capacities
of a quantum channel and some results from \cite{WTB16}. The
secret-key-agreement capacity of a quantum channel is equal to the optimal
rate at which a sender and receiver can use a quantum channel many times, as
well as a free assisting classical channel, in order to establish a reliable
and secure secret key. It is relevant in the context of quantum key
distribution \cite{bb84,SBPC+09}. More generally, since capacity is a limiting
notion that can never be reached in practice, one can consider a fixed
$(n,P^{\leftrightarrow},\varepsilon)$ secret-key-agreement protocol that uses
a channel $n$ times and has $\varepsilon$ error, while generating a secret key
at the rate $P^{\leftrightarrow}$ \cite{WTB16}. Such protocols were explicitly
discussed in \cite{TGW14,TGW14Nat,WTB16}, as well as the related developments
in \cite{Goodenough2015,AML16,Christandl2017,BA17,KW17,KW17a,RKBKMA17,TSW17}.
For a fixed integer $n$ and a fixed error $\varepsilon\in(0,1)$, the
non-asymptotic secret-key-agreement capacity of a quantum channel
$\mathcal{N}$ is written as $P_{\mathcal{N}}^{\leftrightarrow}(n,\varepsilon)$
and is equal to the optimal secret-key rate $P$ subject to these constraints
\cite{WTB16}. That is,%
\begin{equation}
P_{\mathcal{N}}^{\leftrightarrow}(n,\varepsilon)\equiv\sup\{P^{\leftrightarrow
}\ |\ (n,P^{\leftrightarrow},\varepsilon)\text{ is achievable using
}\leftrightarrow\},
\end{equation}
where $\leftrightarrow$ indicates the free use of LOCC between every use of
the quantum channel.

If an $(n,P^{\leftrightarrow},\varepsilon)$ protocol takes place over a
bosonic channel, the stated definition of a secret-key-agreement protocol
makes no restriction on the photon number of the channel input states, and as
such, it is called an unconstrained protocol. The corresponding capacity is
called the unconstrained secret-key-agreement capacity. For example, in such a
scenario, a sender and receiver could freely make use of the following
classically correlated Basel state with mean photon number equal to $\infty$:%
\begin{equation}
\overline{\beta}_{AB}\equiv\frac{6}{\pi^{2}}\sum_{n=1}^{\infty}\frac{1}{n^{2}%
}|n\rangle\langle n|_{A}\otimes|n\rangle\langle n|_{B} ,
\label{eq:basel-state-sep}%
\end{equation}
which represents a dephased version of the state in
\eqref{eq:entangled-Basel-state}. Even though it is questionable whether such
states are physically realizable in practice, they are certainly normalizable,
and thus allowed to be used in principle in an unconstrained
secret-key-agreement protocol.

One of the main results of \cite{WTB16} is the following bound on
$P_{\mathcal{L}_{\eta}}^{\leftrightarrow}(n,\varepsilon)$ when the channel is
taken to be a pure-loss channel $\mathcal{L}^{\eta}$ of transmissivity
$\eta\in(0,1)$:%
\begin{equation}
P_{\mathcal{L}_{\eta}}^{\leftrightarrow}(n,\varepsilon)\leq-\log_{2}%
(1-\eta)+C(\varepsilon)/n, \label{eq:WTB-bnd}%
\end{equation}
where%
\begin{equation}
C(\varepsilon)\equiv\log_{2}6+2\log_{2}(\left[  1+\varepsilon\right]  /\left[
1-\varepsilon\right]  ).
\end{equation}
This bound was established by proving that $P^{\leftrightarrow}\leq-\log
_{2}(1-\eta)+C(\varepsilon)/n$ for any fixed $(n,P^{\leftrightarrow
},\varepsilon)$ unconstrained protocol. As a consequence of this
\textit{uniform} bound, one can then take a supremum over all
$P^{\leftrightarrow}$ such that there exists an $(n,P^{\leftrightarrow
},\varepsilon)$ protocol and conclude \eqref{eq:WTB-bnd}, as was done in the
proof of \cite[Theorem~24]{WTB16}. Similar reasoning was employed in
\cite{WTB16}\ in order to arrive at bounds on the unconstrained capacities of
other bosonic Gaussian channels.

A critical tool used to establish \eqref{eq:WTB-bnd} is the simulation of a
quantum channel via teleportation \cite[Section~V]{BDSW96} (see also
\cite[Theorem 14 \&\ Remark~11]{Mul12}), which, as discussed previously, has
been extended to bosonic states and channels \cite{NFC09,WPG07}, by making use
of the well known bosonic teleportation protocol from \cite{prl1998braunstein}%
. More generally, one can allow for general local operations and classical
communication (LOCC)\ when simulating a quantum channel from a resource state
\cite[Eq.~(11)]{HHH99}, known as LOCC\ channel simulation. This tool is used
to reduce any arbitrary LOCC-assisted protocol over a teleportation-simulable
channel to one in which the LOCC assistance occurs after the final channel
use. Another critical idea is the reinterpretation of a three-party
secret-key-agreement protocol as a two-party private-state generation protocol
and employing entanglement measures such as the relative entropy of
entanglement as a bound for the secret-key rate \cite{HHHO05,HHHO09}. Finally,
one can employ the Chen formula for the relative entropy of Gaussian states
\cite{PhysRevA.71.062320}, as well as a formula from \cite{WTLB16} for the
relative entropy variance of Gaussian states. These tools were foundational
for the results of  \cite[Theorem~24]{WTB16}\ in order to
argue for bounds on secret-key-agreement capacities of bosonic Gaussian channels.

A key point mentioned above, which is critical to and clearly stated in the
proof of \cite[Theorem~24]{WTB16}, is as follows: the proof begins by
considering a fixed $(n,P^{\leftrightarrow},\varepsilon)$ protocol and then
establishes a uniform bound on $P^{\leftrightarrow}$, independent of the
details of the particular protocol.
The discussion given in this paper clarifies that strong convergence in
teleportation simulation suffices for the proof of \cite[Theorem~24]{WTB16}. Furthermore,
 there is no need to invoke the energy-constrained diamond distance
\cite{Sh17,Win17}\ in order to establish the correctness of the proof.

\section{Detailed review of the proof of Theorem 24 in WTB17}

\label{sec:proof-review}

We can now step through the relevant parts of the proof of \cite[Theorem~24]%
{WTB16} carefully, in order to clarify its correctness. It is worthwhile to do so, given that Ref.~\cite{P17} recently questioned the proof of \cite[Theorem~24]%
{WTB16}. I highlight, in
italics, quotations from the proof of \cite[Theorem~24]{WTB16} for clarity and
follow each quotation with a brief discussion. The proof begins by stating that:

\textit{\textquotedblleft First, consider an arbitrary }$(n,P^{\leftrightarrow
},\varepsilon)$\textit{ protocol for the thermalizing channel }$\mathcal{L}%
_{\eta,N_{B}}$\textit{. It consists of using the channel }$n$\textit{~times
and interleaving rounds of LOCC\ between every channel use. Let }$\zeta
_{\hat{A}\hat{B}}^{n}$\textit{ denote the final state of Alice and Bob at the
end of this protocol.\textquotedblright}

It is crucial to note here that this is saying, as it is written, that we
should really start with a particular, fixed $(n,P^{\leftrightarrow
},\varepsilon)$ protocol. This does not mean that we should be considering a
sequence of such protocols
or protocols involving unnormalizable states.
Thus, proceeding by fixing an $(n,P^{\leftrightarrow},\varepsilon)$ protocol,
the proof continues with

\textit{\textquotedblleft By the teleportation reduction procedure [...], such
a protocol can be simulated by preparing }$n$\textit{\ two-mode squeezed
vacuum (TMSV) states each having energy }$\mu-1/2$\textit{ (where we think of
}$\mu\geq1/2$\textit{ as a very large positive real), sending one mode of each
TMSV through each channel use, and then performing continuous-variable quantum
teleportation \cite{prl1998braunstein} to delay all of the LOCC\ operations
until the end of the protocol. Let }$\rho_{\eta,N_{B}}^{\mu}$\textit{ denote
the state resulting from sending one share of the TMSV\ through the
thermalizing channel, and let }$\zeta_{\hat{A}\hat{B}}^{\prime}(n,\mu
)$\textit{ denote the state at the end of the simulation.\textquotedblright}

This means, as it states, that the fixed protocol can be simulated with some
error by replacing each channel use of $\mathcal{L}_{\eta,N_{B}}$\ with the
simulating channel $\mathcal{L}_{\eta,N_{B}}^{\mu}$, for $\mu\in
\lbrack0,\infty)$. Here, the original channel $\mathcal{L}_{\eta,N_{B}}$ is in
correspondence with $\mathcal{G}_{A}$ from Section~\ref{sec:tele-sim},
$\mathcal{L}_{\eta,N_{B}}^{\mu}$ with $\mathcal{G}_{A}^{\bar{\sigma}}$, and
$\mu$ with $\bar{\sigma}$, in the sense that the limit $\mu\to\infty$ is in
correspondence with the limit $\bar{\sigma}\to0$. Furthermore, since this is
an LOCC\ simulation and the original protocol consists of channel uses of
$\mathcal{L}_{\eta,N_{B}}$\ interleaved by LOCC, it is possible to write the
simulating protocol as one that consists of a single round of LOCC on the
state $[\rho_{\eta,N_{B}}^{\mu}]^{\otimes n}$, as observed in
\cite{BDSW96,NFC09,Mul12}. Continuing,

\textit{\textquotedblleft Let $\varepsilon_{\operatorname{TP}}(n,\mu)$ denote
the \textquotedblleft infidelity\textquotedblright\ of the simulation:%
\begin{equation}
\varepsilon_{\operatorname{TP}}(n,\mu)\equiv1-F(\zeta_{\hat{A}\hat{B}}%
^{n},\zeta_{\hat{A}\hat{B}}^{\prime}(n,\mu)).\ \ \ \text{\textquotedblright}%
\end{equation}
}

In the context of the proof, this infidelity clearly corresponds to the
infidelity of the simulation of the fixed protocol. One might argue that the
notation $\varepsilon_{\operatorname{TP}}(n,\mu)$ somehow hides this
dependence on a fixed protocol, but the dependence of $\varepsilon
_{\operatorname{TP}}(n,\mu)$ on a fixed protocol is clear from the context and
the fact that the quantity $1-F(\zeta_{\hat{A}\hat{B}}^{n},\zeta_{\hat{A}%
\hat{B}}^{\prime}(n,\mu))$ itself clearly depends on a fixed protocol as
stated. This infidelity is similar to that in \eqref{eq:adap-err}, except the
initial state of the protocol is constrained to be a separable, unentangled
state shared between the sender and receiver, and each channel $\mathcal{A}%
^{(j)}$ in the protocol is constrained to be an LOCC\ channel between the
sender and receiver. Continuing,

\textit{\textquotedblleft Due to the fact that continuous-variable
teleportation induces a perfect quantum channel when infinite energy is
available \cite{prl1998braunstein}, the following limit holds for every $n$:%
\begin{equation}
\limsup_{\mu\rightarrow\infty}\varepsilon_{\operatorname{TP}}(n,\mu
)=0.\ \ \ \text{\textquotedblright}%
\end{equation}
}

This is indeed a key step for the proof. In the context of the proof given,
the convergence is as it is written and is to be understood in the strong
sense of \eqref{eq:strong-converge-tele-adap}:\ for a \textit{fixed protocol}
used in conjunction with the $n$ teleportation simulations, the infidelity
converges to zero in the limit of ideal squeezing and ideal detection (as
$\mu\rightarrow\infty$ or, equivalently, as $\bar{\sigma}\rightarrow0$). Note
that here we are employing the notion of strong convergence in
\eqref{eq:strong-converge-tele-adap}, given that the original protocol has
LOCC channels interleaved between every use of~$\mathcal{L}_{\eta,N_{B}}$.


Continuing the proof, \textit{\textquotedblleft By using that $\sqrt
{1-F(\rho,\sigma)}$ is a distance measure for states $\rho$ and $\sigma$ (and
thus obeys a triangle inequality) [...], the simulation leads to an
$(n,P^{\leftrightarrow},\varepsilon(n,\mu))$ protocol for the thermalizing
channel, where%
\begin{equation}
\varepsilon(n,\mu)\equiv\min\!\left\{  1,\left[  \sqrt{\varepsilon}%
+\sqrt{\varepsilon_{\operatorname{TP}}(n,\mu)}\right]  ^{2}\right\}  .
\end{equation}
Observe that $\limsup_{\mu\rightarrow\infty}\varepsilon(n,\mu)=\varepsilon$,
so that the simulated protocol has equivalent performance to the original
protocol in the infinite-energy limit.\textquotedblright}

This last part that I have recalled is saying how the error of the simulating
protocol is essentially equal to the sum of the error of the original protocol
and the error from substituting the original $n$ channels with $n$ unideal
teleportations. It is a straightforward consequence of the triangle inequality
for the metric $\sqrt{1-F}$, as recalled there, and thus captures a correct
propagation of errors.

From this point on in the proof, using other techniques,
the following bound is concluded on the rate $P^{\leftrightarrow}$\ of the
fixed $(n,P^{\leftrightarrow},\varepsilon)$\ protocol by invoking the
meta-converse in \cite[Theorem~11]{WTB16}:%
\begin{multline}
P^{\leftrightarrow}\leq D(\rho_{\eta,N_{B}}^{\mu}\Vert\sigma_{\eta,N_{B}}%
^{\mu})+\sqrt{\frac{2V(\rho_{\eta,N_{B}}^{\mu}\Vert\sigma_{\eta,N_{B}}^{\mu}%
)}{n(1-\varepsilon(n,\mu))}}\\
+C(\varepsilon(n,\mu))/n.
\end{multline}
This bound holds for all $\mu$ sufficiently large, $\varepsilon\in(0,1)$, and
positive integers $n$, and as such, it is a uniform bound. Given the uniform
bound, the limit $\mu\rightarrow\infty$ is then taken to arrive at%
\begin{equation}
P^{\leftrightarrow}\leq-\log\!\left(  \left(  1-\eta\right)  \eta^{N_{B}%
}\right)  -g(N_{B})+\sqrt{\frac{2V_{\mathcal{L}_{\eta,N_{B}}}}{n\left(
1-\varepsilon\right)  }}+\frac{C(\varepsilon)}{n}. \label{eq:final-bnd}%
\end{equation}
Since this latter bound is itself a uniform bound, holding for all
$(n,P^{\leftrightarrow},\varepsilon)$ protocols, it is then concluded that%
\begin{multline}
P_{\mathcal{L}_{\eta,N_{B}}}^{\leftrightarrow}(n,\varepsilon)\leq
-\log\!\left(  \left(  1-\eta\right)  \eta^{N_{B}}\right)  -g(N_{B})\\
+\sqrt{\frac{2V_{\mathcal{L}_{\eta,N_{B}}}}{n\left(  1-\varepsilon\right)  }%
}+\frac{C(\varepsilon)}{n}.
\end{multline}
Further arguments are given in the proof of \cite[Theorem~24]{WTB16}\ to
establish the bound in \eqref{eq:WTB-bnd} for the pure-loss channel.

To summarize, the original proof of \cite[Theorem~24]{WTB16} as given there is
correct as it is written. The proof of \cite[Theorem~24]{WTB16} as written
there establishes that for a fixed $(n,P^{\leftrightarrow},\varepsilon)$
protocol, the bound in \eqref{eq:final-bnd}\ holds.
Furthermore,
one might think that a proof may \textit{only} be given by employing the
energy-constrained diamond distance, but this is clearly not the case either,
as demonstrated above.

\section{Conclusion}

\label{sec:conclusion}

The continuous-variable, bosonic quantum teleportation protocol from
\cite{prl1998braunstein} is often loosely stated to simulate an ideal quantum
channel in the limit of infinite squeezing and ideal homodyne detection. The
precise form of convergence is typically not clarified in the literature, and
as a consequence, this has the potential to lead to confusion in mathematical
proofs that employ this protocol.

This paper has clarified various notions of channel convergence, with
applications to the continuous-variable bosonic teleportation protocol from
\cite{prl1998braunstein}, and extended these notions to various contexts. This
paper provided an explicit proof that the continuous-variable bosonic
teleportation protocol from \cite{prl1998braunstein} converges
\textit{strongly} to an ideal quantum channel in the limit of ideal squeezing
and detection. At the same time, this paper proved that this protocol does
\textit{not} converge \textit{uniformly} to an ideal quantum channel, and this
highlights the role of the present paper in providing a precise clarification
of the convergence that occurs in the continuous-variable bosonic
teleportation protocol from \cite{prl1998braunstein}. I also proved that the
teleportation simulations of the pure-loss, thermal, pure-amplifier,
amplifier, and additive-noise channels converge both \textit{strongly} and
\textit{uniformly} to the original channels, which is in contrast to what
occurs in the teleportation simulation of the ideal channel. I suspect that
the explicit uniform bounds in \eqref{eq:loss-unif-bnd},
\eqref{eq:amp-unif-bnd-final}, and \eqref{eq:add-noise-unif-upp-bnd}, on the
accuracy of the teleportation simulations of these channels, will be useful in
future applications. The uniform convergence results were then generalized to the teleportation simulations of particular multi-mode bosonic Gaussian channels. Finally, I gave a physical interpretation of the
convergence results discussed in this paper, by means of the CV Teleportation
Game of Section~\ref{sec:TP-game}.

I also reviewed the proof of \cite[Theorem~24]{WTB16} and confirmed its
correctness as it is written there.
One might think that it is necessary to use the energy-constrained diamond
distance to arrive at a proof of \eqref{eq:WTB-bnd}, but this is clearly not
the case.

It would be interesting in future work to explore physical scenarios of
interest in which different topologies of convergence lead to physically
distinct outcomes, as was the case in the CV Teleportation Game. Recent work
in this direction is available in \cite{Sh17a} and \cite{Win17}.

\begin{acknowledgments}
I acknowledge discussions with several of my dear colleagues, as well as
support from the Office of Naval Research. I am also indebted to an anonymous referee for suggesting the generalization of the uniform convergence of the teleportation simulations of thermal, amplifier, and additive-noise channels to the case of particular multi-mode bosonic Gaussian channels.
\end{acknowledgments}

\bibliographystyle{unsrt}
\bibliography{Ref}

\end{document}